\documentclass[12pt]{amsart}
\usepackage{amsmath,amsfonts,amssymb,amsthm,epsfig,color, graphicx, natbib, appendix,comment,subcaption, graphicx, amssymb,xcolor,geometry,bm,enumerate}

\usepackage{marginnote}

\usepackage{bbm}
\usepackage{setspace} 
\usepackage[colorlinks=true,
    linkcolor=blue,
    filecolor=magenta,      
    urlcolor=blue]{hyperref}
\usepackage{tikz}
\usetikzlibrary{shapes,arrows}
\usetikzlibrary{intersections}
\usetikzlibrary{external}
\usepackage[capitalise,noabbrev,nameinlink]{cleveref}

\newtheorem{theorem}{Theorem}
\newtheorem{lemma}{Lemma}
\newtheorem*{definition}{Definition}
\newtheorem{assumption}{Assumption}

\newtheorem{proposition}{Proposition}

\theoremstyle{definition}
\newtheorem{example}{Example}
\theoremstyle{remark}

\newtheorem{fact}{Fact}
\newtheorem*{example*}{Example}
\theoremstyle{remark}

\usepackage{tikz}
\usetikzlibrary{positioning,quotes}

\definecolor{darkblue}{rgb}{0.0, 0.0, 0.55}
\hypersetup{colorlinks,linkcolor={darkblue},citecolor={darkblue},urlcolor={darkblue}} 
\allowdisplaybreaks
\title{Network Interventions: Targeting Agents or Targeting Links?}

\author{Krishna Dasaratha \and Anant Shah}

\thanks{Dasaratha: krishnadasaratha@gmail.com, Boston University. Shah: anantshah200@gmail.com, University of Cambridge. We are grateful to Ben Golub and Jason Hartline for very valuable discussions. We would also like to thank (in random order) Federico Bobbio, Matthew Jackson, Bruno Strulovici, and Aravindan Vijayaraghavan for helpful comments. We used LLM tools including Refine.ink to check writing and proofs, and Claude assisted with calculations in Section 5. Any errors are our own.}

\begin{document}

\begin{abstract}
    Consider a network game with linear best responses and spillovers between players, and let agents endogenously choose their links. A planner considers interventions to subsidize actions and/or links between players, aiming to maximize a welfare objective. The structure of an optimal intervention is shaped by the intrinsic value that links provide to agents. When this value is non-negative, it is optimal to focus only on subsidizing actions. When it is negative, we give conditions under which it is necessary to include link subsidies. This reverses the basic structure of the optimal intervention in settings with exogenous links. 
    
\end{abstract}

\begin{titlepage}
\date{\today}
\maketitle
\thispagestyle{empty} 
\end{titlepage}

\section{Introduction}
\label{s:introduction}

Many economic settings feature peer effects between connected agents. As one example, consider firms deciding how much to invest in research and development in the presence of a network of R\&D collaborations.\footnote{Research and development collaborations are an important driver of innovation (see, for example, \citealp*{belderbos2004cooperative}). Motivated by this, a theoretical literature models firms' choices about R\&D investment as a network game (see, for example, \citealp{GJ-03}; \citealp{GM-01}; and \citealp{konig2019r}).} Another leading example is students deciding how much to study given their social networks.\footnote{A large literature analyzes peer effects from friends in a social network  (e.g., \citealp{BDF-09} and \citealp{calvo2009peer}). Much of this work treats links as exogenous, but links may be endogenous in practice \citep{goldsmith2013social}.} Following a large literature, we model such settings as network games with complementarities: each agent chooses a level of effort, and an agent's effort is more productive when neighbors exert more effort (\citealp{BCZ-06}).

When network spillovers are important, a natural question is how to target interventions given these spillovers (\citealp*{galeotti2020targeting}). Policymakers often implement interventions incentivizing individual agents. But other types of interventions target links between agents. Consider, for instance, a government trying to subsidize R\&D. One approach is to target individual effort with subsidies or related policies, e.g., subsidizing R\&D at particular firms. Another is to create or subsidize links between agents, e.g., subsidizing collaborative efforts between firms or R\&D clusters. We examine how these two intervention strategies interact and identify conditions under which one is more effective than the other.   

We analyze these questions in a network game of strategic complementarities with endogenous link formation. Agents simultaneously choose a costly action and costly bilateral link intensity for every other agent --- corresponding to an investment to establish connections. These link intensities determine a network of spillovers, where the symmetric strength of a link between a pair of agents is the sum of their link intensities. An important parameter, which we call the \emph{baseline link incentive} (denoted by $s_{ij}$), describes whether a link has intrinsic value to an agent. If this incentive is positive ($s_{ij} > 0$), agent $i$ forms a link even when agent $j$ takes no action while if it is negative ($s_{ij} < 0$), a link is formed only if agent $j$ is exerting enough effort.



A planner intervenes by providing monetary subsidies to agents for taking an action and forming links. The common monetary units let us compare action subsidies to link subsidies. The planner seeks to maximize a welfare objective when agents play an equilibrium in the resulting game, subject to a budget constraint on the subsidy payments. The objective can be any arbitrary increasing function of agents' actions. An example is the sum of all agents' actions, which could correspond to a  policymaker targeting aggregate investment in R\&D or the total time students devote to studying. What will the planner focus on: action subsidies or link subsidies?

Both types of subsidies induce complex spillovers which affect both the equilibrium action profile and the payments the planner must make. To provide an illustration, suppose the planner increases the action subsidy of a given individual. This has a direct effect: the agent increases his action. But this also induces various indirect effects: the agent strengthens his links, this causes his neighbors to increase their actions and links, etc. Similarly, increasing the link subsidy for a given pair of individuals has a number of direct and indirect effects. The overall spillovers are quite involved, but we can nevertheless compare the two types of interventions by exploiting a simple relationship between the spillovers they induce.


Our main results describe whether the planner subsidizes links or actions. It turns out the key determinant is the baseline incentive to form links.
If a pair of agents has a non-negative baseline incentive to form links ($s_{ij} \geq 0$), the optimal policy does not subsidize their link. By contrast, if two connected agents have a negative baseline incentive ($s_{ij} < 0$), then the optimal policy will not subsidize both of their actions without subsidizing their link. Notably, these qualitative properties of the optimal intervention are not very sensitive to how the costs to agents of exerting effort and forming links compare.

The basic insight is that the planner should not subsidize connections that have inherent value to the agents, but may want to subsidize connections between welfare-relevant agents who are not inclined to connect. In the education setting, for example, the result cautions against policies encouraging connections between socially compatible students. Connecting students who would not interact for non-academic reasons, however, can be optimal. This point might be straightforward in a model with binary link choices, but is less obvious when agents choose the intensity of connections.

To give some intuition for the results, if the network is sufficiently connected, then (as we describe below) the benefits to strengthening links can be very large. Moreover, the link formation protocol creates free-riding incentives between pairs of agents. But it is nevertheless cheapest for the planner to reward high actions and let agents form stronger links in response. When baseline incentives are negative, however, agents' underinvestment in link formation can become more severe, so the planner can optimally target this margin by directly subsidizing links. 

We next compare the endogenous link formation model to a benchmark model with exogenous links and find the basic structure of the optimal intervention reverses. In this benchmark, agents make link-formation decisions without considering the subsequent network game (which is the relevant object for the planner). This is essentially equivalent to a setting where links are exogenous, and we choose this formulation so that the planner's interventions to strengthen links can be implemented via subsidies. We show that when the baseline incentives to form links are large enough, it is optimal to only subsidize links. A key reason is that when the network is connected enough, the returns to further increasing connections grow very large. The result is the opposite of our main finding in the endogenous case, where the planner optimally subsidizes only actions when these baseline incentives are non-negative.


Finally, we ask to what extent our main result generalizes beyond quadratic settings. We identify conditions on agents' utilities driving each part of our result. A main insight is that the results are not sensitive to changing agents' costs of forming links or taking actions. So the main determinant of whether to subsidize links or actions is not which is more expensive for agents.


\subsection{Related Literature}

There are large literatures on network games in economics, computer science, and related fields (see \citealp{menache2011network} and \citealp{jackson2015games} for overviews). We begin with a framework with linear best responses, which is widely used for its tractability (e.g., \citealp*{BCZ-06} and \citealp*{bramoulle2014strategic}), and focus on the case when actions are complementary. This model has been applied to a wide range of empirical settings, including our main examples above of education (\citealp*{calvo2009peer}) and research and development (\citealp*{konig2019r}). Our endogenous network model is closest to parts of this literature incorporating network formation (e.g., \citealp{sadler2021games}).

More recent work has explored optimal interventions in network games. Most of this literature focuses on targeting individual behavior via subsidies for effort, prices, or related instruments (e.g., \citealp*{candogan2012optimal}, \citealp*{bloch2016targeting}, \citealp*{galeotti2020targeting}, \citealp*{parise2023graphon}). A smaller group of papers focuses on interventions targeting links (\citealp*{belhaj2016efficient} and \citealp{li2023designing}) and gives conditions when the optimal networks belong to a particular class: nested split graphs. We allow both types of interventions and focus on the tradeoffs between the two.

Most closely related, \cite*{sun2023structural} and \cite*{kor2025welfare} also consider interventions when the planner can target individuals and links. The main result in \cite{sun2023structural} is an equivalence between the two types of interventions: any intervention targeting only links induces the same outcomes as some intervention targeting only individuals. \cite{kor2025welfare} extend the tools from \cite{galeotti2020targeting} to settings where the planner can apply their budget to subsidizing actions or changing links. This gives a characterization of the optimal intervention in terms of a spectral decomposition of the adjacency matrix of the original network.\footnote{Another difference from the present work is a quadratic cost to interventions, including the monetary subsidies for actions. We consider a more traditional budget constraint and find this leads to starker structure for the optimal intervention.} We focus on asking when one instrument dominates the other. Because the planner directly changes the network in prior work, such comparisons would be determined by a parameter capturing the relative costs of the two instruments. We instead allow endogenous network formation, which lets us compare the value of paying for links and actions in common monetary units, and find it is often optimal to use only one of the instruments.

A bit further afield, there are long lines of work on optimal interventions in networks beyond network games. Other applications include diffusion of information (e.g., \citealp*{kempe03maximizeinfluence}, \citealp*{banerjee2013diffusion}, and \citealp{campbell2013word}), opinion dynamics (e.g., \citealp*{yildiz2011discrete} and \citealp{sadler2023influence}), and multi-agent contracts (\citealp*{dasaratha2024incentive} and \citealp*{CMO}).


\section{Model}
\label{s:model}

We analyze a network game with endogenous link formation. Consider a set of $n$ agents, who we denote by $\mathcal{N} = \{1,2,\dots,n\}$, and a social planner. (We will use the pronouns ``he" to refer to an agent and ``she" to refer to a planner.)

Each agent $i$ chooses a real-valued action $a_i \geq 0$ and a real-valued link intensity $g_{ij} \geq 0$ with each other agent $j \in \mathcal{N} \setminus \{i\}$. The link intensities between $i$ and $j$ determine the strength $G_{ij} = G_{ji} = g_{ij} + g_{ji}$ of a weighted, undirected link between $i$ and $j$. The collection of links can be represented as a non-negative matrix $\bm{G} \in \mathbb{R}_{\geq 0}^{n \times n}$, with zero entries $G_{ii}=0$ on the diagonal.

The utility to agent $i$ from action profile $\bm{a}$ and link intensities $\bm{g}$ is \begin{multline*}
U_{i}\left(\bm{a},\bm{g}\right) = b_{i}a_{i} + \rho \sum_{j \in \mathcal{N} \setminus \{i\}} G_{ij}a_{i}a_{j}
- \frac{a_{i}^2}{2} + \sum_{j \in \mathcal{N} \setminus \{i\}}s_{ij}G_{ij} - \sum_{j \in \mathcal{N} \setminus \{i\}} \frac{f_{ij}g_{ij}^2}{2}.
\end{multline*}
We describe each term in the utility function:
\begin{itemize}
    \item The first term is a linear individual payoff $b_i \in \mathbb{R}_{>0}$ to actions  which we call the standalone baseline incentive for the action. 
    \item The second term captures spillovers between agents' actions. The parameter $\rho \in \mathbb{R}_{> 0}$ describes the global strength of complementarities. Each summand $G_{ij}a_ia_j$ implies that there are complementarities between agents $i$ and $j$, and these complementarities are stronger when the link $G_{ij}$ is larger.
    \item The third term is a quadratic cost $(a_{i}) = a_{i}^2/2$ of actions.\footnote{We could allow heterogeneous costs $C_{i}(a_{i}) = c_{i}a_{i}^2/2$, but it is without loss of generality to normalize the coefficients $c_i$ to $1$.}
    \item The fourth term is an intrinsic payoff $s_{ij} \in \mathbb{R}$ to links, which we call the standalone baseline incentive for a link. 
    \item The final term is a quadratic cost $F_{ij}(g_{ij}) = f_{ij}g_{ij}^2/2$ to link intensities, where the coefficient $f_{ij} \in \mathbb{R}_{> 0}$ can differ across $i$ and $j$.
\end{itemize} 

We assume that the intrinsic payoff from links $s_{ij}$ and the curvature of link intensity costs $f_{ij}$ are symmetric across pairs, that is, $s_{ij}=s_{ji}$ and $f_{ij}=f_{ji}$. A natural explanation for this symmetry is that the inherent values of links depend only on pairwise characteristics, such as similarity between agents or the distance separating them. The model allows the various standalone incentives and costs to be heterogeneous across agents and across links.  


The planner maximizes an objective  function $\mathcal{W}(\bm{a})$ depending on agents' actions, and this objective is strictly increasing and differentiable in each of its arguments. She can intervene by providing linear subsidies rewarding agents' actions and links. Specifically, agent $i$ receives a subsidy $\beta_{i} \geq 0$ for exerting effort. Furthermore, agent $i$ receives a subsidy $\sigma_{ij} \geq 0$ for forming a link $G_{ij}$ with agent $j$. For a fixed intervention $\left(\bm{\beta},\bm{\sigma}\right)$ and realized profile $(\bm{a},\bm{g})$, agent $i$'s utility is \begin{align*}
    \mathcal{U}_{i} = U_{i}\left(\bm{a},\bm{g}\right) + \beta_{i}a_{i} + \sum_{j \in \mathcal{N} \setminus \{i\}} \sigma_{ij}G_{ij}. 
\end{align*} The total payment made by the planner is \begin{align*}
    \mathcal{P} = \sum_{i \in \mathcal{N}} \beta_{i}a_{i} + \sum_{i \in \mathcal{N}}\sum_{j \in \mathcal{N} \setminus \{i\}} \sigma_{ij}G_{ij}.
\end{align*} The planner has a budget $\mathcal{B} > 0$, which can be allocated between action and link subsidies.

Once the planner commits to an intervention, agents respond by playing a pure strategy Nash equilibrium $\left(\bm{a}^*,\bm{g}^*\right)$. If no equilibrium exists under a particular subsidy scheme, we set the welfare $\mathcal{W}$ under that subsidy scheme to $-\infty$. We analyze the planner's problem under a given equilibrium selection $(\bm{a}^*(\bm{\beta},\bm{\sigma}), \bm{g}^*(\bm{\beta},\bm{\sigma}))$. Her problem is therefore to maximize the objective $\mathcal{W}(\cdot)$ under this selection, subject to a budget constraint: \begin{align*}
\sup_{\left(\bm{\beta},\bm{\sigma}\right)} & \quad \mathcal{W}(\bm{a^*}), \\
    \text{subject to} & \quad \left(\bm{a}^*,\bm{g}^*\right) =(\bm{a}^*(\bm{\beta},\bm{\sigma}), \bm{g}^*(\bm{\beta},\bm{\sigma})), \\
    & \quad \mathcal{P}(\bm{a}^*,\bm{g}^*,\bm{\beta},\bm{\sigma}) \leq \mathcal{B}.
\end{align*} 

For this problem to be well-defined, there must be a subsidy and a corresponding equilibrium that satisfy the budget constraint. A complication is that an equilibrium need not exist if the global complementarity parameter or the planner's subsidies are sufficiently large. A standard fixed-point argument shows that for the zero subsidy $\bm{\beta}=0$ and $\bm{\sigma}=0$, an equilibrium exists as long as the complementarity parameter $\rho$ is small enough (see \Cref{fact:eqlbf} in \Cref{a:omittedproofs}). The restriction on the global complementarity parameter plays a similar role to the spectral radius restriction in standard linear-quadratic network games (e.g., \citealp{BCZ-06}).

\subsection{Remarks on the model} \label{ss:modelremarks}



We interpret the subsidies $(\bm{\beta},\bm{\sigma})$ as monetary payments to agents, and this lets us compare action and link subsidies in common monetary units. We assume that subsidies target bilateral links $G_{ij}$ and not agents' link intensities $g_{ij}$, e.g., because the planner observes the strength of the link but not who is investing to sustain the link. While we require payments to depend on the weight $G_{ij}$ of the undirected link, we do not restrict the link subsidy to be symmetric across a pair of agents. The planner can target agents differently due to heterogeneity.

We note that it is not important for our analysis whether the intrinsic value $s_{ij}$ of links is a coefficient on the link $G_{ij}$, the link intensity $g_{ij}$, or some combination of the two (as the difference is a pure externality term that does not change equilibrium).

Once the planner selects a subsidy $(\bm{\beta},\bm{\sigma})$, agents choose actions $\bm{a}^*$ and link intensities $\bm{g}^*$ simultaneously. Studying a model with real-valued links and simultaneous link and action choices lets us study Nash equilibrium. We can therefore analyze optimal interventions via a first-order approach. Alternative solution concepts commonly used in the strategic link-formation literature, such as pairwise stability \citep{jackson1996strategic}, would necessitate different methods.

\section{Spillover Effects Of Policy Interventions}
\label{s:spilloverdynamics}

This section describes the welfare effects of increasing subsidies to actions or links. We will see there is a clean relation between the benefits to the planner from each of these types of subsidies.


Before analyzing the effects of perturbing subsidies, we begin with a characterization of equilibrium under a fixed subsidy scheme:
\begin{proposition}
\label{p:eqlbendo}
    Consider an intervention $(\bm{\beta},\bm{\sigma})$ and a pure strategy Nash equilibrium $(\bm{a}^*,\bm{g}^*)$ under this intervention. Then actions satisfy \begin{align}\label{eq:action_FOC}
       \left(\bm{I} - \rho \bm{G}^*\right)\bm{a}^*  = \bm{b} + \bm{\beta}.
    \end{align}

    Furthermore, any link intensity $g_{ij}^*$ for which the equilibrium first-order condition binds satisfies \begin{align}\label{eq:link_FOC}
        g_{ij}^* = \frac{\rho a_{i}^*a_{j}^* + s_{ij} + \sigma_{ij}}{f_{ij}}. 
    \end{align}
\end{proposition}

Given equilibrium link choices, the characterization of equilibrium actions in (\ref{eq:action_FOC}) matches \cite*{BCZ-06}. The main difference, of course, is that links are chosen endogenously in our model, and (\ref{eq:link_FOC}) describes equilibrium link intensities. We will now use this characterization to compute the welfare effects of perturbing a subsidy. 

A helpful insight is that closely related spillover dynamics govern the effects of changing action and link subsidies, despite these instruments targeting different objects.

\begin{proposition}
    \label{l:welfarederivativeendogenous}
    Suppose the equilibrium selection $(\bm{a}^*(\bm{\beta},\bm{\sigma}), \bm{g}^*(\bm{\beta},\bm{\sigma}))$ is differentiable in a neighborhood of $(\bm{\beta},\bm{\sigma})$. If $g_{ij}^*>0$, then the derivatives of equilibrium welfare in $\beta_{i}$, $\beta_j$, and $\sigma_{ij}$ are related by 
    \begin{align*}
        \frac{d\mathcal{W}}{d \sigma_{ij}} = \frac{\rho}{f_{ij}} \left( \frac{d\mathcal{W}}{d \beta_{i}} a_{j}^* + \frac{d\mathcal{W}}{d \beta_{j}} a_{i}^* \right).
    \end{align*}
\end{proposition}

We now provide some intuition for this result. First suppose the planner marginally increases the subsidy for $i$'s action. What is the effect on welfare? The direct effect is that agent $i$ responds by increasing his action $a_i$. This then has various spillovers: agent $i$ now increases his link intensities $g_{ij}$ with other agents; these agents increase their action $a_j$ and link intensity $g_{ji}$ to agent $i$. Each of these changes also induces further spillovers, etc. The derivative $\frac{d\mathcal{W}}{d \beta_{i}}$ captures the overall change in equilibrium welfare due to these cascading effects.

Now suppose the planner instead increases the subsidy for $i$'s link to $j$. Then agent $i$ invests more in his bilateral link intensity to $j$. Since the overall link between $i$ and $j$ is stronger, they increase their actions $a_i$ and $a_j$. These in turn have spillover effects as in the previous paragraph, but the key point is that these spillovers induce the same total change in welfare. So we can write $\frac{d\mathcal{W}}{d\sigma_{ij}}$ as a linear combination of $\frac{d\mathcal{W}}{d \beta_{i}}$ and $\frac{d\mathcal{W}}{d \beta_{j}}$. This basic insight also appears in \cite{li2023designing} where a planner can intervene to change exogenous links, and the proposition points out that it extends to our model with endogenous links.

The proposition compares the benefits to the planner from action and link subsidies, but the full trade-off of course also depends on the costs of the two instruments. Increasing a subsidy imposes the cost of the new subsidy but also induces spillovers that make existing subsidies more expensive. In the next section, we will analyze optimal interventions by  comparing both the benefits and costs of subsidies.

\section{Optimal Interventions With Endogenous Spillovers}
\label{s:endonetworkform}

This section asks whether the optimal subsidy targets actions or links. Our main result identifies a region where it is optimal to focus only on subsidizing actions and a complementary region where including link subsidies can be optimal.


Our analysis relies on two regularity conditions at the optimal subsidy scheme. First, we assume that the equilibrium strategy is differentiable with respect to small changes in the subsidy. Because the endogenous link formation model may admit multiple equilibria for a given subsidy scheme, this assumption requires that, at the optimum, the selected equilibrium varies smoothly with local perturbations in subsidies. Second, we assume that the planner's budget constraint binds. This condition rules out perverse cases in which an increase in subsidies induces all agents to reduce both their actions and link intensities. These perverse cases correspond to unstable equilibria, and one could alternately refine our equilibrium concept to require stability.

\begin{assumption} \label{as:regularity}
   At the optimal subsidy $\left(\bm{\beta}^*,\bm{\sigma}^*\right)$, the following conditions hold: \begin{enumerate}
       \item the equilibrium selection $(\bm{a}^*(\bm{\beta},\bm{\sigma}),\bm{g}^*(\bm{\beta},\bm{\sigma}))$ is differentiable in a neighborhood of $(\bm{\beta}^*,\bm{\sigma}^*)$, and, 
       \item the planner's budget constraint strictly binds.
   \end{enumerate} 
\end{assumption}

We can now present our main results. We find that the qualitative structure of an intervention is governed by the baseline incentive $s_{ij}$ to form a link. Recall that this baseline incentive is the marginal value to agent $i$ of increasing $G_{ij}$ from zero if all agents' actions are zero.

\begin{theorem} \label{t:endogenousmodel}
    Let $\left(\bm{\beta}^*,\bm{\sigma}^*\right)$ be an optimal subsidy satisfying \Cref{as:regularity}. If $s_{ij} \geq 0$, the optimal intervention does not subsidize the link between $i$ and $j$: $\sigma_{ij}^*=\sigma_{ji}^*=0$.
\end{theorem}

The theorem states that when agents have a positive baseline incentive to form links ($s_{ij} \geq 0$), the planner does not subsidize link formation and focuses on action subsidies. This holds despite various sources of underinvestment in links, including agents' incentives to free-ride on bilateral links. 

What if the baseline link incentive $s_{ij}$ is negative? In this case, we can give sufficient conditions for the optimal intervention to subsidize the link $G_{ij}$.

\begin{proposition} \label{p:negativebaselineendo}
Let $\left(\bm{\beta}^*,\bm{\sigma}^*\right)$ be an optimal subsidy satisfying \Cref{as:regularity}. Consider any pair of agents $i$ and $j$ such that $s_{ij} < 0$ and $g_{ij}^*,g_{ji}^*>0$. If $\beta_{i}^*>0$ and $\beta_{j}^*>0$, then  $\sigma_{ij}^* + \sigma_{ji}^*>0$.
\end{proposition}

If the baseline incentive to form a link is negative ($s_{ij} < 0$), then the planner does not subsidize the actions of two connected agents $i$ and $j$ without also subsidizing their link.

Crucially, both these comparisons do not depend on parameters that may seem more relevant. Of particular note, the result is independent of the coefficients governing the cost of link formation, $f_{ij}$. For instance, the optimal policy does not subsidize link formation when $s_{ij} \geq 0$ even if $f_{ij}$ is much smaller than the curvature of the cost of actions. Moreover, the characterization is also robust to the choice of welfare objective.

We next describe the intuition underlying these results. Perturbing a link subsidy can generate substantially larger welfare effects than perturbing an action subsidy (see \Cref{l:welfarederivativeendogenous}). But when baseline incentives are non-negative ($s_{ij} \geq 0$),  subsidizing actions is a strictly cheaper way of inducing proportional spillovers in equilibrium (at the optimal subsidy). Stronger links may be very beneficial, but it turns out to be optimal for the planner to subsidize actions and let agents form these links endogenously. In contrast, when baseline incentives are negative ($s_{ij} < 0$), agents underinvest in link formation, and the planner optimally offsets this distortion by subsidizing links.

To formalize this, recall when agents play strategy profile $(\bm{a}^*,\bm{g}^*)$, the total payments made by the planner are: \begin{align*}
    \mathcal{P} = \sum_{i \in \mathcal{N}}\beta_{i}a_{i}^* + \sum_{i \in \mathcal{N}} \sum_{j \in \mathcal{N} \setminus \{i\}}\sigma_{ij}G_{ij}^*.
\end{align*} By differentiating this expression and manipulating terms, we obtain a relation between the cost of perturbing an action subsidy and the cost of perturbing a link subsidy (when link intensities $g_{ij}^*,g_{ji}^*>0$): \begin{align} \label{eq:paymentderivativeendogenous}
        \frac{d\mathcal{P}}{d \sigma_{ij}} - \frac{\rho}{f_{ij}} \left( \frac{d\mathcal{P}}{d \beta_{i}} a_{j}^* + \frac{d\mathcal{P}}{d \beta_{j}} a_{i}^* \right) = \underbrace{\frac{2\left(s_{ij} + \sigma_{ij} + \sigma_{ji}\right)}{f_{ij}}}_{\mathcal{D}_{ij}}.
    \end{align}

The intuition for (\ref{eq:paymentderivativeendogenous}) is very similar to the intuition for \Cref{l:welfarederivativeendogenous}, which analogously established a relation between the welfare effect of perturbing an action and a link subsidy. For a comparison involving planner payments, the key difference is the additional term $\mathcal{D}_{ij}$ on the right-hand side. When the planner perturbs a link subsidy, she needs to pay for the existing link level $G_{ij}^*$ in addition to payments due to changes in equilibrium actions. The term $\mathcal{D}_{ij}$ accounts for change in pay involving the existing link.  

When $s_{ij} \geq 0$, the right-hand side $\mathcal{D}_{ij}$ is non-negative. This implies that a subsidy to link formation is strictly more costly than an action subsidy yielding the same benefit, making action subsidies the preferred instrument. In contrast, when $s_{ij} < 0$, equilibrium link intensities under zero link subsidies are small. In this case, the planner optimally exploits the low marginal cost of inducing link formation by subsidizing links. 

\medskip

\Cref{t:endogenousmodel} provides conditions ensuring the optimal intervention only subsidizes actions. One might wonder whether this is because the objective $\mathcal{W}(\mathbf{a})$ depends only on actions and not links. We next give an example where the optimal intervention  includes action subsidies even when the planner's objective depends only on links. So the optimal policy does not merely subsidize the type of effort valued by the planner; rather, \Cref{t:endogenousmodel} captures a meaningful difference between links and actions.


\begin{example}\label{ex:link_objective}
    In this example, we modify our welfare function to depend on links rather than actions. Suppose the welfare function \begin{align*}
    \mathcal{W}(\bm{G}) = \sum_{i \in \mathcal{N}} \sum_{j \in \mathcal{N} \setminus \{i\}}G_{ij}
\end{align*} is the sum of link weights.

We consider a planner budget $\mathcal{B} = 0.01$. Individual incentives to take actions and form links are uniformly sampled as \begin{align*}
  b_{i} & \sim U([0,\overline{b}]) \quad \text{ for every } i \in \mathcal{N}, \\ \text{and} \quad s_{ij} & \sim U([0,\overline{s}]) \text{ for every pair } i,j \in \mathcal{N}, 
\end{align*} with $\overline{b}=\overline{s}=0.01$. The cost of link formation $f_{ij}= 10$ for every pair $i$ and $j$, and the cost of action is $10$ for every agent $i$. Finally, we set the parameter $\rho$ controlling the strength of complementarities to $1$. We numerically solve for the optimal subsidy in this environment, selecting the equilibrium $(\bm{a}^*,\bm{g}^*)$ reached by starting with $(\bm{a},\bm{g}) = (\bm{0},\bm{0})$ and iteratively taking best responses. \Cref{fig:subsidyperformanceratio} illustrates the role of subsidizing actions in an optimal intervention.

\begin{figure}
    \centering
    \includegraphics[width=0.5\linewidth]{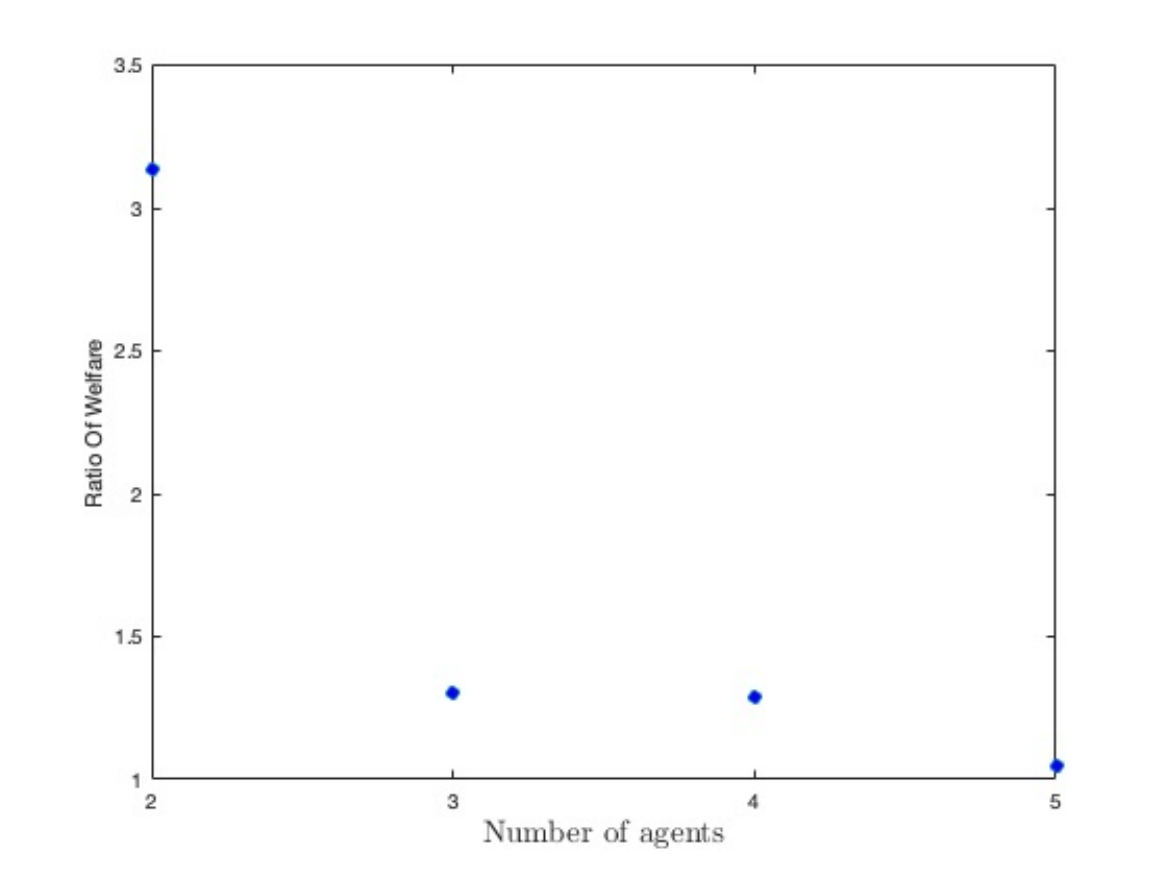}
    \caption{Ratio of welfare of the optimal subsidy to the welfare of the optimal subsidy which only subsidizes links as a function of the number of agents.}
    \label{fig:subsidyperformanceratio}
\end{figure}

The figure plots the average of the ratio of welfare under the optimal subsidy to welfare under the best subsidy restricted to links only, where the average is taken over $20$ samples. A ratio exceeding one indicates that the optimal intervention necessarily includes action subsidies. We see that only subsidizing links can be far from optimal. This reinforces the role of action subsidies: they generate spillovers at a lower cost than link subsidies when link formation is endogenous. 


\end{example}

\section{Optimal Interventions In A Benchmark Model}
\label{s:exonetworkform}

We next compare our results on endogenous link formation in \Cref{s:endonetworkform} with a benchmark model where links are essentially exogenous: agents choose links prior to the network game and without considering its payoffs. In this benchmark, the qualitative structure of the optimal intervention is reversed. In particular, when baseline incentives to form links are high, optimal interventions only subsidize links and not actions.

 
We now describe the benchmark model. The timing is different: agents first publicly choose link intensities $g_{ij}$ and then choose actions $a_i$. The link intensities $\bm{g}$ determine a symmetric network of spillovers $G_{ij} = G_{ji} = g_{ij} + g_{ji}$ as before. Each agent $i$ chooses each link intensity $g_{ij}$ to maximize 
$$ s_{ij} G_{ij} + \sigma_{ij} G_{ij} - \frac{f_{ij}g_{ij}^2}{2},$$
where $s_{ij}$ is a linear benefit from links and $\sigma_{ij}$ is a subsidy parameter chosen by the planner. The linear benefit from links $s_{ij}$ and the curvature of link intensity costs $f_{ij}$ is symmetric across a pair of agents: $s_{ij} = s_{ji}$ and $f_{ij} = f_{ji}$. After links realize, agents play a standard network game (with payoffs specified below).

Importantly, this benchmark assumes agents do not consider payoffs in the subsequent network game when choosing links. This may be a reasonable model in settings where agents' primary motivations for forming links are different from the outcome of interest to the planner. For example, in some settings students may form network ties for primarily social reasons while policymakers may care about educational outcomes affected by those networks.

The planner can also choose action subsidies $\beta_i$, and player $i$'s utility in the intervention game is  $$\mathcal{U}_{i}(\bm{a},\bm{g}) = (b_{i}+\beta_i)a_{i} + \rho \sum_{j \in \mathcal{N} \setminus \{i\}}G_{ij}a_{i}a_{j} - \frac{a_{i}^2}{2}.$$
Here $b_i > 0$ captures an agent's individual incentive to exert effort in the absence of subsidies.

As before, given a subsidy scheme, agents play a pure strategy Nash equilibrium $\left(\bm{a}^*,\bm{g}^*\right)$ and we fix an equilibrium selection $(\bm{a}^*(\bm{\beta},\bm{\sigma}), \bm{g}^*(\bm{\beta},\bm{\sigma}))$. If an equilibrium does not exist given some subsidy scheme we set $\mathcal{W}$ to $-\infty$, and as before we can rule out non-existence by taking the costs of links to be sufficiently large. The planner chooses an intervention $(\bm{\beta},\bm{\sigma})$ to maximize a welfare function $\mathcal{W}(\bm{a}^*)$, which remains differentiable and strictly increasing in each argument, subject to a budget $\mathcal{B}$.

We want to characterize the optimal intervention when baseline incentives to form links are sufficiently high. To do so, we consider a sequence of environments indexed by $m \in \mathbb{N}$. Along this sequence, we assume that all incentives $b_{i}$ lie in an interval $\left[\underline{b},\overline{b}\right]$ with $\overline{b} \geq \underline{b} > 0$, all cost coefficients $f_{ij}$ lie in an interval $\left[\underline{f},\overline{f}\right]$ with $\overline{f} \geq \underline{f} > 0$, the complementarity parameter $\rho$ is bounded away from zero, and the budgets are bounded above. For each $m$, we define the \emph{zero-subsidy network} $\mathbf{G}_m^{(0)}$ to be the equilibrium network in the absence of subsidies. We will analyze the case when the zero-subsidy networks have high spillovers, as captured by the following assumption:

\begin{assumption} \label{as:highspillovers} The sequence of zero-subsidy networks $\{\bm{G}_{m}^{(0)}\}_{m \in \mathbb{N}}$ satisfies:
    \begin{enumerate}
        \item (Uniform Primitivity) There exist $L \in \mathbb{N}$ and $\delta>0$ such that for every matrix $\bm{G}_{m}^{(0)}$ in the sequence, \begin{align*}
       \frac{\left(\bm{G}_{m}^{(0)}\right)^{L}}{\left(\mu_{1}\left(\bm{G}_{m}^{(0)}\right)\right)^{L}} \geq \delta \bm{1} \bm{1}^{\top}
    \end{align*} entrywise, where $\mu_{1}\left(\bm{G}^{(0)}_{m}\right)$ is the spectral radius of $\bm{G}^{(0)}_{m}$. 
    \item The spectral radius of $\rho \bm{G}_{m}^{(0)}$ converges to $1$:\begin{align*}
        \lim_{m \to \infty} \mu_{1}\left(\rho\bm{G}_{m}^{(0)}\right) = 1.
    \end{align*}
    \end{enumerate}
\end{assumption}


We briefly interpret each condition. For part (a), an unweighted graph is defined to be primitive if there exists an integer $L$ such that there is a walk of length $L$ between any two agents. The condition extends this to weighted graphs by requiring that the total weights of such walks, measured by the appropriate entry of $\left(\bm{G}_{m}^{(0)}\right)^{L}$, be sufficiently large. We moreover require that the integer $L$ can be chosen uniformly across the sequence. Part (b) requires the spectral radius of the sequence of zero-subsidy networks, adjusted to account for the complementarity parameter, converges to $1$. As this spectral radius approaches $1$, spillovers and therefore equilibrium actions grow very large.

The main result of this section establishes that in such a high spillover regime, the planner does not subsidize actions and only focuses on link subsidies:

\begin{theorem} \label{t:benchmarkmodelquadratic}
    Consider any sequence of environments and corresponding sequence of zero-subsidy networks $\{\bm{G}_{m}^{(0)}\}_{m \in \mathbb{N}}$ satisfying \Cref{as:highspillovers}. For $m$ large enough, any optimal subsidy $\left(\bm{\beta}^*,\bm{\sigma}^*\right)$ has $\beta_{i}^*=0$ for every agent $i \in \mathcal{N}$.
\end{theorem}


This result stands in sharp contrast to the structure of an optimal subsidy scheme in the endogenous network formation model (see \Cref{t:endogenousmodel}). In that model, the planner does not subsidize link formation when the baseline incentives $s_{ij}$ are positive.

We now provide some intuition for the theorem. An agent's first-order condition is \begin{align*}
    a_i = (b_i + \beta_i) + \rho \sum_{j \in \mathcal{N}\backslash \{i\} }G_{ij}a_j.
\end{align*} In the high-spillover regime, equilibrium actions grow large. So link subsidies become very valuable due to the terms $G_{ij}a_j$ in first-order conditions. Moreover the  costs of link subsidies remain bounded as $m \rightarrow \infty$. As actions grow large, on the other hand, subsidies to actions become very costly. While the benefits to action subsidies may also grow, the proof confirms that link subsidies must give a better cost-benefit tradeoff.

A caveat to the theorem, which limits its scope, is that as $m \rightarrow \infty$, the budget $\mathcal{B}$ must vanish for an optimal intervention to exist. This is because the planner can implement arbitrarily large actions if the budget is large enough to induce equilibrium networks $\bm{G}^*$ with spectral radius at least $\frac{1}{\rho}$. One can provide necessary and sufficient conditions for subsidizing a link outside of this limit. Since these conditions involve endogenous quantities, however, we focus on the cleaner limit result.



\section{Optimal Subsidies In General Environments}\label{s:generalization}

The analysis thus far has relied on quadratic functional form assumptions. We now ask what structure drives each part of our main result. To do so, we describe a more general model and provide sufficient conditions for analogues of \Cref{t:endogenousmodel} and \Cref{p:negativebaselineendo}.

Agents continue to take real-valued actions $a_{i} \geq 0$. Each agent $i$ also continues to choose link intensities $g_{ij} \geq 0$ with other agents $j \in \mathcal{N} \setminus \{i\}$. The link intensities $g_{ij}$ and $g_{ji}$ determine the strength $G_{ij}=G_{ji}=g_{ij}+g_{ji}$ of an undirected link between $i$ and $j$. The utility of agent $i$ from action profile $\bm{a}$ and link intensities $\bm{g}$ is \begin{align*}
U_{i}\left(\bm{a},\bm{g}\right) = b_{i}a_{i} + \rho \sum_{j \in \mathcal{N} \setminus \{i\}}G_{ij}u_{ij}(a_{i}){u}_{ji}(a_{j}) - C_{i}(a_{i}) \\
     + \sum_{j \in \mathcal{N} \setminus \{i\}} s_{ij}G_{ij} - \sum_{j \in \mathcal{N} \setminus \{i\}}F_{ij}(g_{ij}). 
\end{align*} We now describe each term in the utility function:

\begin{itemize}
    \item The standalone marginal benefit from taking an action is $b_{i} \in \mathbb{R}_{> 0}$. The cost is given by a twice differentiable, strictly increasing and strictly convex function $C_{i} : \mathbb{R}_{\geq 0} \to \mathbb{R}_{\geq 0}$. The marginal cost at zero action is zero: $C_{i}'(0)=0$.
    \item The standalone baseline incentive to form a link is $s_{ij} \in \mathbb{R}$ and the link incentive is symmetric across pairs: $s_{ij} = s_{ji}$. The cost is given by a twice differentiable, strictly increasing, and strictly convex function $F_{ij} : \mathbb{R}_{\geq 0} \to \mathbb{R}_{\geq 0}$. The marginal cost at a zero link intensity is zero, i.e., $F_{ij}'(0)=0$. The link intensity cost is symmetric across a pair of agents: $F_{ij}(\cdot) = F_{ji}(\cdot)$.
    \item Spillovers depend on individual actions: they are determined by the actions of $i$ and $j$ through twice differentiable and strictly increasing functions $u_{ij} : \mathbb{R}_{\geq 0} \to \mathbb{R}_{\geq 0}$. The spillovers are zero for a zero action: $u_{ij}(0)=0$.
\end{itemize}
We can recover the baseline model by taking all $u_{ij}(a_i)$ to be the identity, the costs $C_i(a_i) = a_i^2/2$, and $F_{ij}(g_{ij}) = f_{ij} g_{ij}^2/2$.


If the planner chooses subsidies $\left(\bm{\beta},\bm{\sigma}\right)$, then agent $i$'s utility is \begin{align*}
    \mathcal{U}_{i} = U_{i}\left(\bm{a},\bm{g}\right) + \beta_{i}a_{i} + \sum_{j \in \mathcal{N} \setminus \{i\}} \sigma_{ij}G_{ij}. 
\end{align*} Once the planner commits to an intervention, agents play a pure strategy Nash equilibrium $(\bm{a}^*,\bm{g}^*)$. The assumptions on equilibrium existence and equilibrium selection are identical to those in \Cref{s:model}. 

To extend our results beyond the linear-quadratic setting, we will introduce several properties of functions. We begin with the costs of forming links. 

\begin{definition}
    A function $F(\cdot)$ is \begin{enumerate}
        \item super-quadratic, if $GF''(G) \geq F'(G)$ for all  $G \geq 0$ and $F'(G)$ is convex.
        \item sub-quadratic, if $GF''(G) \leq F'(G)$ for all $G \geq 0$.
    \end{enumerate}
\end{definition}

To provide intuition for these conditions, suppose that $F_{ij}(\cdot)$ takes the form $$
    F_{ij}(g) = f_{ij}g^{\gamma}$$ for some positive $f_{ij}$. Then $F_{ij}(\cdot)$ is super-quadratic if $\gamma \geq 2$ and  sub-quadratic if $\gamma \in (1,2]$.
    
    We next introduce properties on the function $u_{ij}(\cdot)$, which governs the role of actions in spillovers. 

\begin{definition}
    A function $u(\cdot)$ is \begin{enumerate}
        \item super-linear, if $u(a) \leq au'(a)$ for all $a \geq 0$.
        \item sub-linear, if 
            $u(a) \geq au'(a)$ for all $a \geq 0$.
    \end{enumerate}
\end{definition}

A simple example satisfying these conditions is $$u_{ij}(a) = \omega_{ij} a^{\kappa}$$
for $\omega_{ij}>0$.
In this case, the function $u_{ij}(\cdot)$ is super-linear if $\kappa \geq 1$ and sub-linear if $\kappa \in (0,1]$.

We can now present our generalization of \Cref{t:endogenousmodel}:

\begin{theorem}
    \label{t:baselineincentivegeneral}
  Let $\left(\bm{\beta}^*,\bm{\sigma}^*\right)$ be an optimal subsidy satisfying \Cref{as:regularity}. Suppose $s_{ij} \geq 0$. If
        \begin{enumerate}[(a)]
            \item the link formation cost functions $F_{ij}(\cdot)$ are super-quadratic, and
            \item the payoff functions $u_{ij}(\cdot)$ are sub-linear,
        \end{enumerate} then the optimal link subsidy is $\sigma_{ij}^*=\sigma_{ji}^*=0$.
\end{theorem} The theorem states that when $s_{ij} \geq 0$, and link formation costs grow sufficiently steeply, and spillovers from agents' actions grow sufficiently slowly, the optimal subsidy does not subsidize link formation. The intuition is that, under steep costs, link intensities respond weakly to marginal changes in link subsidies. At the same time, slowly growing spillovers dampen the return from strengthening existing links. Accordingly, the planner focuses on action subsidies. A surprising implication of the result is that the functional form of the costs of actions can be an arbitrary convex function. In particular, the cost of actions can be much steeper than the cost of link formation.

We can also generalize \Cref{p:negativebaselineendo}:

 \begin{proposition} \label{p:negativebaselineincentivegeneral}
  Let $\left(\bm{\beta}^*,\bm{\sigma}^*\right)$ be an optimal subsidy satisfying \Cref{as:regularity}. Suppose $s_{ij} < 0$. If the following are satisfied: \begin{enumerate}[(a)]
            \item the link formation cost function $F_{ij}(\cdot)$ is sub-quadratic, and
            \item the payoff functions $u_{ij}(\cdot)$ are super-linear, and
            \item the optimal action subsidies satisfy $\beta_i^*,\beta_j^*>0$, and the link intensities $g_{ij}^*,g_{ji}^*>0$,
        \end{enumerate} then the optimal link subsidy has $\sigma_{ij}^* + \sigma_{ji}^*>0$.
\end{proposition}This result says that when $s_{ij} < 0$, and link formation costs do not grow too quickly and action complementarities grow steeply enough, the planner will not subsidize two connected agents' actions without subsidizing links. To provide some intuition, agents underinvest in link formation when $s_{ij} < 0$. With link formation costs that are not too steep, link intensities are more responsive to marginal changes in link subsidies. Furthermore, convex spillovers propagate the return from strengthening existing links. 

The results in this section generalize \Cref{t:endogenousmodel} and \Cref{p:negativebaselineendo} to broader classes of functional forms. Unlike in the quadratic case, the two parts of the theorem are no longer complementary as \Cref{t:baselineincentivegeneral} and \Cref{p:negativebaselineincentivegeneral} require different assumptions on curvature. Their purpose is instead to clarify the forces underlying each result.

\section{Conclusion} \label{s:conclusion}

We have considered a planner intervening in a network game with endogenous links, and compared subsidies to actions and links. When baseline incentives to form links are non-negative, the optimal intervention only subsidizes actions. The benefits to links can be very high, but it is most cost effective to subsidize actions and let agents form links. When baseline incentives to form links are negative, we give sufficient conditions for the optimal intervention to include link subsidies. Each of these results continues to hold for a wider non-parametric class of utility functions. 

We find that the structure of interventions is reversed when links are exogenous. In particular, optimal interventions include link subsidies when spillovers are large as induced by a non-negative baseline incentive to form links.

We view the contribution of this paper as being twofold: first, highlighting the primitives of the environment that govern the qualitative properties of optimal subsidies and second, demonstrating the structural difference in subsidies between endogenous and exogenous links. These insights suggest that a policymaker  intervening in network settings with peer effects benefits from carefully tailoring the intervention to agents' motivations to form links.  

\bibliographystyle{ecta}
\bibliography{references}

\appendix

\section{Omitted Proofs}
\label{a:omittedproofs}

\begin{fact}
\label{fact:eqlbf} There exists a $\widetilde{\rho} > 0$ such that if $\rho < \widetilde{\rho}$, then a Nash equilibrium exists under the zero subsidy $\left(\bm{\beta},\bm{\sigma}\right) = (\bm{0},\bm{0})$. 
\end{fact}

\begin{proof} Consider the subsidy $\left(\bm{\beta},\bm{\sigma}\right) = (\bm{0},\bm{0})$ and any agent $i \in \mathcal{N}$. Fix actions and links of other agents denoted by $\left(\bm{a}_{-i},\bm{g}_{-i}\right)$. Suppose the agents optimize actions $a_{i}$ and link intensities $g_{ij}$ individually. We will establish the existence of a fixed-point of the agents' best-response correspondence. The utility of an agent is strictly concave in $a_{i}$ and the unique optimal action is \begin{align*}
        a_{i}^* = b_{i} + \rho \sum_{j \in \mathcal{N} \setminus \{i\}}G_{ij}a_{j}. 
    \end{align*} Furthermore, an agent's utility is strictly concave in link intensity $g_{ij}$. The unique optimal link intensity is \begin{align} \label{eq:optlinkintensity}
        g_{ij}^* = \max \left \{0,\frac{\left(\rho a_{i}a_{j} + s_{ij} \right)}{f_{ij}} \right \}. 
    \end{align} 
    
    Thus, it suffices to establish a fixed point of the best-response map $\Phi : \mathbb{R}^{n}_{\geq 0} \to \mathbb{R}_{\geq 0}^{n}$ defined as \begin{align*}
        \left(\Phi(\bm{a})\right)_{i} := b_{i} + \rho \sum_{j \in \mathcal{N} \setminus \{i\}} G_{ij}a_{j} \quad \text{ for every } i \in \mathcal{N}, 
    \end{align*} where $G_{ij} = g_{ij} + g_{ji}$ with $g_{ij}$ defined as in (\ref{eq:optlinkintensity}). We define $\overline{s} := \max_{i,j}|s_{ij}|$ and $\overline{b} := \max_{i}b_{i}$. Consider the following upper bound on the best-response of any agent: \begin{align*}
        \left(\Phi(\bm{a})\right)_{i} & \leq \overline{b} + \rho \sum_{j \in \mathcal{N} \setminus \{i\}}|G_{ij}|a_{j}, \\
        & \leq \overline{b} + \rho \sum_{j \in \mathcal{N} \setminus \{i\}} \frac{2\rho a_{i} a_{j} + 2\overline{s}}{f_{ij}} a_{j}. \\
    \end{align*} Given curvatures $f_{ij}$, define $\underline{f} := \min_{i,j}f_{ij}$. Suppose there exists a $R>0$ such that \begin{align} \label{eq:Rbound}
        \overline{b} + \frac{2 \rho^2}{\underline{f}}nR^3 + \frac{2 \rho \overline{s}}{\underline{f}}nR < R.
    \end{align} Then, applying Brouwer's fixed point theorem on the compact and convex domain $[0,R]^{n}$, the continuous map $\Phi(\cdot)$ has a fixed point in $[0,R]^{n}$. When $R = 2\overline{b}$ and \begin{align*}
        \rho < \frac{\sqrt{n^2 \overline{s}^2 + 4 \underline{f} n \overline{b}^2} - n\overline{s}}{8n\overline{b}^2},
    \end{align*} equation (\ref{eq:Rbound}) is satisfied. We conclude that there is a fixed point $\bm{a}^{*}$ such that action $a_{i}^*$ is at most $R$ for every agent $i \in \mathcal{N}$.

    We now show that individual optimization of actions and links indeed gives a Nash equilibrium. Consider the fixed point $\bm{a}^*$ established in the paragraph above and corresponding link intensities $\bm{g}^*$. It suffices to show that the utility $\mathcal{U}_{i} : \mathbb{R}_{\geq 0}^{n} \to \mathbb{R}$ is strictly concave. Here, the utility is defined holding action and links of other agents $\left(\bm{a}_{-i}^*,\bm{g}_{-i}^*\right)$ fixed. The Hessian of $\mathcal{U}_{i}$ is negative definite if and only if \begin{align*}
        \rho^2 \sum_{j \in \mathcal{N} \setminus \{i\}}\frac{\left(a_{j}^*\right)^2}{f_{ij}} < 1.
    \end{align*} We have established in the paragraph above that $a_{j}^* \leq 2\overline{b}$ for every $j \in \mathcal{N}$. Using $f_{ij} \geq \underline{f}$, if $\rho < \sqrt{\underline{f}} / \left(2\overline{b } \sqrt{n-1}\right)$, the utility is concave. Combining with the fixed point construction in the paragraph above, an equilibrium exists for every $$\rho < \left\{\frac{\sqrt{\underline{f}}}{2\overline{b}\sqrt{n-1}},\frac{\sqrt{n^2 \overline{s}^2 + 4 \underline{f} n \overline{b}^2} - n\overline{s}}{8n\overline{b}^2}\right\}.$$ The statement of the result follows. 
\end{proof}

\begin{proof}[Proof of \Cref{l:welfarederivativeendogenous}]
    Consider any subsidy scheme $(\bm{\beta},\bm{\sigma})$ and a corresponding equilibrium $(\bm{a}^*,\bm{G}^*)$ that is differentiable in a neighborhood of $(\bm{\beta},\bm{\sigma})$. Applying \Cref{p:eqlbendo}, for every $j \in \mathcal{N}$, the equilibrium action satisfies \begin{align} \label{eq:actionfoc}
        a_{j}^* = b_{j} + \beta_{j} + \rho \sum_{k \in \mathcal{N} \setminus \{j\}}G_{jk}a_{k}^*.
    \end{align} Furthermore, for any individual $k \in \mathcal{N}$ whose link incentive constraint binds with agent $l \in \mathcal{N}$, the equilibrium link intensity satisfies \begin{align} \label{eq:linkfoc}
       f_{kl} g_{kl}^* = \sigma_{kl} + s_{kl} + \rho a_{k}^*a_{l}^* .
    \end{align} We will implicitly differentiate (\ref{eq:actionfoc}) and (\ref{eq:linkfoc}) with respect to $\beta_{i}$ and $\sigma_{ij}$ respectively to solve for the derivative of equilibrium actions. We begin with a perturbation in action subsidy $\beta_{i}$. Implicitly differentiating (\ref{eq:actionfoc}) with respect to $\beta_{i}$, we have \begin{align} \label{eq:actionderbetajneqi}
        \frac{d a_{j}^*}{d \beta_{i}} = \rho \left( \sum_{k \in \mathcal{N} \setminus \{j\}}G_{jk}^* \frac{d a_{k}^*}{d \beta_{i}} + \sum_{k \in \mathcal{N} \setminus \{j\}}\left(\frac{d g_{jk}^*}{d \beta_{i}} + \frac{d g_{kj}^*}{d \beta_{i}}\right) a_{k}^* \right).
    \end{align} Furthermore, for agent $i$ we have \begin{align} \label{eq:actionderbetajeqi}
         \frac{d a_{i}^*}{d \beta_{i}} = 1 + \rho \left(\sum_{k \in \mathcal{N} \setminus \{i\}}G_{ik}^* \frac{d a_{k}^*}{d \beta_{i}} + \sum_{k \in \mathcal{N} \setminus \{i\}}\left(\frac{d g_{ik}^*}{d \beta_{i}} + \frac{d g_{ki}^*}{d \beta_{i}}\right) a_{k}^*\right).
    \end{align} Now, implicitly differentiating (\ref{eq:linkfoc}) with respect to $\beta_{i}$ for every link intensity from agent $k$ to agent $l$ whose incentive constraint binds, we have \begin{align} \label{eq:linkderivbeta}
        f_{kl}\frac{d g_{kl}^*}{d \beta_{i}} = \rho \left(a_{k}^* \frac{d a_{l}^*}{d \beta_{i}} + \frac{d a_{k}^*}{d \beta_{i}} a_{l}^*\right).
    \end{align} For any link intensity constraint that does not bind, we have $d g_{kl}^*/d \beta_{i} = 0$. We now substitute the link intensity derivatives in (\ref{eq:actionderbetajneqi}), for every $j \in \mathcal{N} \setminus \{i\}$. In the calculation below, we assume all link intensity incentive constraints bind. When they do not, the index of the summation and coefficients are appropriately defined over links for which constraints bind. Thus, we have \begin{align} 
         \frac{d a_{j}^*}{d \beta_{i}} &= \rho \left( \sum_{k \in \mathcal{N} \setminus \{j\}}G_{jk}^* \frac{d a_{k}^*}{d \beta_{i}} + \rho \sum_{k \in \mathcal{N} \setminus \{j\}}\frac{2}{f_{jk}}\left(a_{j}^*\frac{d a_{k}^*}{d \beta_{i}} + a_{k}^*\frac{d a_{j}^*}{d \beta_{i}}\right) a_{k}^* \right), \\ \label{eq:jneqibetafin}
        &= \rho \left(\sum_{k \in \mathcal{N} \setminus \{j\}}\left(G_{jk}^* + \frac{2\rho}{f_{jk}}a_{j}^*a_{k}^*\right)\frac{d a_{k}^*}{d \beta_{i}} + 2\rho \frac{da_{j}^*}{d \beta_{i}}\sum_{k \in \mathcal{N} \setminus \{j\}}\frac{\left(a_{k}^*\right)^2}{f_{jk}}\right). 
    \end{align} Analogously, substituting (\ref{eq:linkderivbeta}) in (\ref{eq:actionderbetajeqi}), we have \begin{align} \label{eq:jeqibetafin}
        \frac{d a_{i}^*}{d \beta_{i}} = 1 + \rho \left(\sum_{k \in \mathcal{N} \setminus \{i\}}\left(G_{ik}^* + \frac{2\rho}{f_{ik}} a_{i}^*a_{k}^*\right)\frac{d a_{k}^*}{d \beta_{i}} + 2\rho \frac{da_{i}^*}{d \beta_{i}}\sum_{k \in \mathcal{N} \setminus \{i\}}\frac{\left(a_{k}^*\right)^2}{f_{ik}}\right).
    \end{align} We combine (\ref{eq:jeqibetafin}) and (\ref{eq:jneqibetafin}) to write the derivative of actions $\bm{a}^*$ with respect to $\beta_{i}$ in vector form. Thus, when agent $i$'s action incentive constraint binds, \begin{align*}
        \left(\bm{I}-\rho \bm{\mathcal{M}}\right)\frac{d \bm{a}^*}{d \beta_{i}} = \bm{e}_{i},
    \end{align*} where $\bm{I}$ is the identity matrix, $\bm{\mathcal{M}}$ is a matrix with entry $\mathcal{M}_{ij} = G_{ij}^* + 2\rho a_{i}^*a_{j}^* / f_{ij}$ and $\mathcal{M}_{ii} = 2\rho  \sum_{k \in \mathcal{N} \setminus \{i\}}  \left(a_{k}^*\right)^2 / f_{ik}$, and $\bm{e}_{i}$ is the $i^{\text{th}}$ basis vector. By assumption, equilibrium $(\bm{a}^*,\bm{g}^*)$ is differentiable at subsidy $(\bm{\beta},\bm{\sigma})$ and thus \begin{align*}
        \frac{d \bm{a}^*}{d \beta_{i}} = \left(\bm{I}-\rho \bm{\mathcal{M}}\right)^{-1}_{i},
    \end{align*} where $\left(\bm{I}-\rho \bm{\mathcal{M}}\right)^{-1}_{i}$ is the $i^{\text{th}}$ column of matrix $\left(\bm{I}-\rho \bm{\mathcal{M}}\right)^{-1}$. The welfare effect of perturbing subsidy $\beta_{i}$ is \begin{align*}
        \frac{d \mathcal{W}}{d \beta_{i}} = \nabla \mathcal{W}^{\top} \frac{d \bm{a}^*}{d \beta_{i}} = \nabla \mathcal{W}^{\top} \left(\bm{I}-\rho \bm{\mathcal{M}}\right)^{-1}_{i}. 
    \end{align*}
    Call this quantity $\mathcal{R}_i$. 

    
We now follow the same set of calculations to solve for the derivative of equilibrium actions $\bm{a}^*$ with respect to link subsidy $\sigma_{ij}$. For every agent $k \in \mathcal{N}$, implicitly differentiating (\ref{eq:actionfoc}) with respect to $\sigma_{ij}$, we have \begin{align} \label{eq:actiondersigmajneqi}
         \frac{d a_{k}^*}{d \sigma_{ij}} = \rho \left( \sum_{l \in \mathcal{N} \setminus \{k\}}G_{kl}^* \frac{d a_{l}^*}{d \sigma_{ij}} + \sum_{l \in \mathcal{N} \setminus \{k\}}\left(\frac{d g_{kl}^*}{d \sigma_{ij}} + \frac{d g_{lk}^*}{d \sigma_{ij}}\right) a_{l}^* \right),
    \end{align}

We implicitly differentiate (\ref{eq:linkfoc}) with respect to $\sigma_{ij}$ for agents with link intensity constraints that bind. For every link intensity $(k,l)$ that is not $(i,j)$, we have \begin{align} \label{eq:linkderiva}
       f_{kl} \frac{d g_{kl}^*}{d \sigma_{ij}} = \rho \left(a_{k}^* \frac{d a_{l}^*}{d \sigma_{ij}} + \frac{d a_{k}^*}{d \sigma_{ij}} a_{l}^*\right),
    \end{align} and for the link intensity from agent $i$ to $j$, we have \begin{align} \label{eq:linkderivb}
       f_{ij} \frac{d g_{ij}^*}{d \sigma_{ij}} = 1 + \rho \left(a_{i}^* \frac{d a_{j}^*}{d \sigma_{ij}} + \frac{d a_{j}^*}{d \sigma_{ij}} a_{i}^*\right).
    \end{align} For link incentive constraints that do not bind, we have $d g_{kl}^* / d \sigma_{ij} = 0$. We now substitute (\ref{eq:linkderiva}) and (\ref{eq:linkderivb}) in (\ref{eq:actiondersigmajneqi}). As we did when taking derivatives with respect to $\beta_{i}$, we assume all link intensities bind in the calculation below. When they do not, the index of summation and coefficients are appropriately defined over links for which the constraint binds. The index and coefficients are identical when taking derivatives with respect to $\beta_{i}$ and $\sigma_{ij}$. Thus, we have \begin{align*}
         \frac{da_{k}^*}{d \sigma_{ij}} = \rho \left(\sum_{l \in \mathcal{N} \setminus \{k\}}\left(G_{kl}^* + \frac{2\rho}{f_{kl}} a_{k}^*a_{l}^*\right)\frac{d a_{l}^*}{d \sigma_{ij}} + 2\rho \frac{d a_{k}^*}{d \sigma_{ij}}\sum_{l \in \mathcal{N} \setminus \{k\}}\frac{\left(a_{l}^*\right)^2}{f_{kl}}\right),  
    \end{align*} for every agent $k \neq i,j$. Furthermore, for agent $i$ and $j$, we have \begin{align*}
         \frac{da_{i}^*}{d \sigma_{ij}} = \frac{\rho}{f_{ij}} a_{j}^* + \rho \left(\sum_{l \in \mathcal{N} \setminus \{i\}}\left(G_{il}^* + \frac{2\rho}{f_{il}} a_{i}^*a_{l}^*\right)\frac{d a_{l}^*}{d \sigma_{ij}} + 2\rho \frac{d a_{i}^*}{d \sigma_{ij}}\sum_{l \in \mathcal{N} \setminus \{i\}}\frac{\left(a_{l}^*\right)^2}{f_{il}}\right), 
    \end{align*} and \begin{align*}
         \frac{da_{j}^*}{d \sigma_{ij}} = \frac{\rho}{f_{ij}} a_{i}^* + \rho \left(\sum_{l \in \mathcal{N} \setminus \{j\}}\left(G_{jl}^* + \frac{2\rho}{f_{jl}} a_{j}^*a_{l}^*\right)\frac{d a_{l}^*}{d \sigma_{ij}} + 2\rho \frac{d a_{j}^*}{d \sigma_{ij}}\sum_{l \in \mathcal{N} \setminus \{j\}}\frac{\left(a_{l}^*\right)^2}{f_{jl}}\right).
    \end{align*} Suppose the link intensity incentive constraint from agent $i$ to agent $j$ binds. We combine the set of equations above to write the derivative of actions $\bm{a}^*$ with respect to $\sigma_{ij}$ in vector form. Thus,  \begin{align*}
        \left(\bm{I}-\rho \bm{\mathcal{M}}\right)\frac{d \bm{a}^*}{d \sigma_{ij}} = \frac{\rho}{f_{ij}} \left(a_{j}^*\bm{e}_{i} + a_{i}^* \bm{e}_{j}\right). 
    \end{align*} By assumption, equilibrium $(\bm{a}^*,\bm{g}^*)$ is differentiable at subsidy $(\bm{\beta},\bm{\sigma})$ and thus \begin{align*}
        \frac{d \bm{a}^*}{d \sigma_{ij}} = \frac{\rho}{f_{ij}} \left(a_{j}^* \left(\bm{I}-\rho \bm{\mathcal{M}}\right)^{-1}_{i} + a_{i}^* \left(\bm{I}-\rho \bm{\mathcal{M}}\right)^{-1}_{j}\right).
    \end{align*} Using this expression, the welfare effect of perturbing subsidy $\sigma_{ij}$ is \begin{align*}
        \frac{d \mathcal{W}}{d \sigma_{ij}} = \frac{\rho}{f_{ij}}\left(\mathcal{R}_{i}a_{j}^* + \mathcal{R}_{j}a_{i}^*\right).
    \end{align*}
    The statement of the result follows. 
\end{proof}


\begin{proof}[Proofs  of \Cref{t:endogenousmodel} and \Cref{p:negativebaselineendo}]
   For a given welfare objective $\mathcal{W}(\cdot)$, the Lagrangian of the planner's optimization problem is given by \begin{align*}
        \mathcal{L}(\bm{a},\bm{g},\bm{\beta},\bm{\sigma},\lambda) = \mathcal{W}(\bm{a}) - \lambda \left(\sum_{k \in \mathcal{N}}\beta_{k}a_{k} + \sum_{k \in \mathcal{N}}\sum_{l \in \mathcal{N} \setminus \{k\}}\sigma_{kl}G_{kl}\right).
    \end{align*} We begin by taking the derivative of the Lagrangian with respect to $\beta_{i}$: \begin{align} \label{eq:lagderivbeta}
        \frac{d \mathcal{L}}{d \beta_{i}} = \nabla \mathcal{W}^{\top} \frac{d \bm{a}^*}{d \beta_{i}} -\lambda \left(\sum_{k \in \mathcal{N}} \beta_{k} \frac{d a_{k}^*}{d \beta_{i}} + \sum_{k \in \mathcal{N}}\sum_{l \in \mathcal{N} \setminus \{k\}}\sigma_{kl}\frac{d G_{kl}^*}{d \beta_{i}}\right) - \lambda a_{i}^*.
    \end{align} Substituting (\ref{eq:linkderivbeta}) in the expression above, there is a vector $\bm{\omega}$ such that the derivative of the Lagrangian with respect to $\beta_{i}$ satisfies \begin{align} \label{eq:reducedvec}
        \frac{d \mathcal{L}}{d \beta_{i}} = \bm{\omega}^{\top} \frac{d \bm{a}^*}{d \beta_{i}} - \lambda a_{i}^*.
    \end{align}  Substituting the expression \begin{align*}
        \frac{d \bm{a}^*}{d \beta_{i}} = \left(\bm{I}-\rho \bm{\mathcal{M}}\right)^{-1}_{i},
    \end{align*} derived in the proof of \Cref{l:welfarederivativeendogenous}, into (\ref{eq:reducedvec}) we have \begin{align}
        \frac{d \mathcal{L}}{d \beta_{i}} &= \bm{\omega}^{\top} \left(\bm{I}-\rho \bm{\mathcal{M}}\right)^{-1}_{i}  - \lambda a_{i}^*, \\ \label{eq:betafocfin}
        &= \mathcal{R}_{i}' - \lambda a_{i}^*,
    \end{align} where we define $\mathcal{R}_{i}' := \bm{\omega}^{\top} \left(\bm{I}-\rho \bm{\mathcal{M}}\right)^{-1}_{i}$.

We next take the derivative of the Lagrangian with respect to $\sigma_{ij}$: \begin{align} \label{eq:lagderivsigma}
        \frac{d \mathcal{L}}{d \sigma_{ij}} = \nabla \mathcal{W}^{\top} \frac{d \bm{a}^*}{d \sigma_{ij}} -\lambda \left(\sum_{k \in \mathcal{N}} \beta_{k} \frac{d a_{k}^*}{d \sigma_{ij}} + \sum_{k \in \mathcal{N}}\sum_{l \in \mathcal{N} \setminus \{k\}}\sigma_{kl}\frac{d G_{kl}^*}{d \sigma_{ij}}\right) - \lambda G_{ij}^*.
    \end{align} When $s_{ij} \geq 0$, the link intensity incentive constraint between agents binds. Furthermore, when $s_{ij} < 0$, we assume that $g_{ij}^*>0$ and $g_{ji}^*>0$. Thus, for the analysis below, suppose the link intensity constraint from agent $i$ to $j$ and from agent $j$ to $i$ binds. 
    
    We substitute (\ref{eq:linkderiva}) and (\ref{eq:linkderivb}) into (\ref{eq:lagderivsigma}) and use $$G_{ij}^* = \frac{\left(2s_{ij} + \sigma_{ij} + \sigma_{ji} + 2\rho a_{i}^*a_{j}^*\right)}{f_{ij}}.$$ Consider the following chain of equations: \begin{align}
        \frac{d \mathcal{L}}{d \sigma_{ij}} &= \bm{\mathcal{\omega}}^{\top} \frac{d \bm{a}^*}{d \sigma_{ij}} - \lambda \frac{\left(\sigma_{ij} + \sigma_{ji}\right)}{f_{ij}} - \lambda G_{ij}^*, \\
        &= \bm{\mathcal{\omega}}^{\top} \frac{d \bm{a}^*}{d \sigma_{ij}} - 2\lambda \frac{\left(\sigma_{ij} + \sigma_{ji}\right)}{f_{ij}} - 2\lambda \frac{\left(s_{ij} + \rho a_{i}^*a_{j}^*\right)}{f_{ij}}, \\ \label{eq:linkfocfin}
        &= \bm{\mathcal{\omega}}^{\top} \frac{d \bm{a}^*}{d \sigma_{ij}} - \frac{2\lambda \rho}{f_{ij}} a_{i}^*a_{j}^* - \frac{2\lambda}{f_{ij}} s_{ij} - 2\lambda \frac{\left(\sigma_{ij}+\sigma_{ji}\right)}{f_{ij}}.
    \end{align}

Consider any optimal subsidy scheme $(\bm{\beta}^*,\bm{\sigma}^*)$ and corresponding equilibrium $(\bm{a}^*,\bm{g}^*)$ such that the budget constraint strictly binds, that is, $\lambda > 0$. Then,
\begin{align}
        \frac{d \mathcal{L}}{d \sigma_{ij}} 
        &= \frac{\rho}{f_{ij}} \left(\mathcal{R}_{i}'a_{j}^* + \mathcal{R}_{j}'a_{i}^* - 2\lambda a_{i}^*a_{j}^*\right) - \frac{2\lambda}{f_{ij}} s_{ij} - 2\lambda \frac{\left(\sigma_{ij}+\sigma_{ji}\right)}{f_{ij}}. \label{eq:linkfocfin2}
    \end{align} Furthermore, analyzing (\ref{eq:betafocfin}), the KKT first-order conditions imply \begin{align*}
    \mathcal{R}_{i}' \leq \lambda a_{i}^* \quad \text{ and } \quad \mathcal{R}_{j}' \leq \lambda a_{j}^*.
\end{align*} Substituting in (\ref{eq:linkfocfin2}), we have \begin{align*}
    \frac{d \mathcal{L}}{d \sigma_{ij}} \leq -\frac{2\lambda}{f_{ij}} s_{ij} - 2 \lambda \frac{\left(\sigma_{ij}+\sigma_{ji}\right)}{f_{ij}}. 
\end{align*} An analogous calculation shows that the derivative of the Lagrangian with respect to $\sigma_{ji}$ satisfies \begin{align*}
    \frac{d \mathcal{L}}{d \sigma_{ji}} \leq -\frac{2\lambda}{f_{ij}} s_{ij} - 2 \lambda \frac{\left(\sigma_{ij}+\sigma_{ji}\right)}{f_{ij}}. 
\end{align*} Thus, if $s_{ij} \geq 0$, it must be that $\sigma_{ij}^* = \sigma_{ji}^*=0$. Otherwise, the total subsidy across the pair of agents $\sigma_{ij}^*+\sigma_{ji}^*>0$, implying that the Lagrangian has a negative derivative with respect to the positive link subsidy, a contradiction. 

In contrast, suppose $s_{ij} < 0$ and $\beta_{i}^*>0$ and $\beta_{j}^*>0$. By the KKT first-order conditions, $\mathcal{R}_{i}'=\lambda a_{i}^*$ and $\mathcal{R}_{j}'=\lambda a_{j}^*$. Substituting in (\ref{eq:linkfocfin}), we have \begin{align*}
    \frac{d \mathcal{L}}{d \sigma_{ij}} = -\frac{2 \lambda}{f_{ij}}s_{ij} - 2 \lambda \frac{\left(\sigma_{ij}+\sigma_{ji}\right)}{f_{ij}}. 
\end{align*} If $\sigma_{ij}^*+\sigma_{ji}^*=0$, then $d \mathcal{L} / d \sigma_{ij} > 0$ which contradicts the KKT first-order conditions. Thus, it must be that the total link subsidy $\sigma_{ij}^*+\sigma_{ji}^*>0$. The statement of the result follows. 
\end{proof}

\begin{proof}[Proof of \Cref{t:benchmarkmodelquadratic}]

We introduce some notation that will be used throughout the proof. Fix a subsidy $\left(\bm{\beta},\bm{\sigma}\right)$ and equilibrium $\left(\bm{a}^*,\bm{g}^*\right)$. For any link intensity $g_{kl}$ such that the equilibrium first-order condition binds, we have \begin{align*}
    f_{kl} g_{kl}^* = s_{kl} + \sigma_{kl}.
\end{align*} Furthermore, by the first-order conditions for $a_i$, equilibrium actions $\bm{a}^*$ must solve \begin{align}
    \label{eq:eqlbexoi} \left(\bm{I}-\rho \bm{G}^*\right) \bm{a}^* = \bm{b} + \bm{\beta}
\end{align} where $\bm{I}$ is the identity matrix. We will establish that the spectral radius of $\bm{G}^*$, which we denote by $\mu_{1}(\bm{G}^*)$, is at most $1/\rho$. The network $\bm{G}^*$ is non-negative. Thus, by the Perron-Frobenius theorem, the spectral radius $\mu_{1}(\bm{G}^*)$ is an eigenvalue of $\bm{G}^*$ with a corresponding eigenvector $\bm{v}_{1}$, which we normalize so $\|v_1\|_2=1$, having all entries non-negative. Multiplying (\ref{eq:eqlbexoi}) by $\bm{v}_1^{\top}$, we have \begin{align*}
    (1-\rho \mu_{1}(\bm{G}^*))\bm{v}_{1}^{\top} \bm{a}^* = \bm{v}^{\top}_{1} \left(\bm{b} + \bm{\beta}\right). 
\end{align*} In the equation above, because $b_{i} + \beta_{i} \geq \underline{b}$ for every $i \in \mathcal{N}$ and $\bm{v}_{1}$ has unit $l_{2}$-norm, the right-hand side is strictly positive. Thus, it must be $\mu_{1}(\bm{G}^*) < 1/\rho$. Consequently, equilibrium actions are \begin{align} \label{eq:eqlbexo}
    \bm{a}^* =  \left(\bm{I}-\rho \bm{G}^*\right)^{-1} \left(\bm{b} + \bm{\beta}\right).
\end{align} We finish by making some useful observations due to this characterization. First, any subsidy $\left(\bm{\beta},\bm{\sigma}\right)$ admits an equilibrium if and only if the spectral radius $\mu_{1}(\bm{G}^*) < 1/\rho$ and this equilibrium is unique. Second, the equilibrium is differentiable in the action subsidy $\beta_{i}$. Finally, when the link intensity $g_{ij}^*>0$, the equilibrium is differentiable in the link subsidy $\sigma_{ij}$.

We now analyze the derivative of equilibrium actions $\bm{a}^*$ with respect to $\beta_{i}$. A straightforward calculation establishes that \begin{align} \label{eq:derivexobeta}
    \frac{d \bm{a}^*}{d \beta_{i}} = \left(\bm{I}-\rho \bm{G}^*\right)^{-1}_{i},
\end{align} where $\left(\bm{I}-\rho \bm{G}^*\right)^{-1}_{i}$ is the $i^{\text{th}}$ column of the matrix $\left(\bm{I}-\rho \bm{G}^*\right)^{-1}$. A similar calculation establishes that the derivative of equilibrium actions with respect to $\sigma_{ij}$, when $g_{ij}^*>0$, satisfies \begin{align} \label{eq:derivexosigma}
    \frac{d \bm{a}^*}{d \sigma_{ij}} = \frac{\rho}{f_{ij}} \left(\left(\bm{I}-\rho \bm{G}^*\right)^{-1}_{i}a_{j}^* + \left(\bm{I}-\rho \bm{G}^*\right)^{-1}_{j}a_{i}^*\right).
\end{align} For the next paragraph, we introduce notation to denote the derivative of welfare and planner payments with respect to $\beta_{i}$ and $\sigma_{ij}$ assuming $g_{ij}^*>0$. For a derivative with respect to $\beta_{i}$, we denote \begin{align} \label{eq:bderivexo}
   \frac{d \mathcal{W}(\bm{a}^*)}{d \beta_{i}} = \underbrace{ \nabla \mathcal{W}^{\top} \left(\bm{I}-\rho \bm{G}^*\right)^{-1}_{i}}_{\mathcal{S}_{i}} \quad \text{ and } \quad \frac{d \mathcal{P}(\bm{a}^*)}{d \beta_{i}} = \underbrace{\bm{\beta}^{\top} \left(\bm{I}-\rho \bm{G}^*\right)^{-1}_{i}}_{\mathcal{T}_{i}} + a_{i}^*.
\end{align} Using (\ref{eq:derivexosigma}), for a derivative with respect to $\sigma_{ij}$, we have \begin{align} \label{eq:sderivexo}
    \frac{d \mathcal{W}(\bm{a}^*)}{d \sigma_{ij}} = \frac{\rho}{f_{ij}} \left(\mathcal{S}_{i}a_{j}^* + \mathcal{S}_{j}a_{i}^*\right) \quad \text{ and } \quad  \frac{d \mathcal{P}(\bm{a}^*)}{d \sigma_{ij}} = \frac{\rho}{f_{ij}} \left(\mathcal{T}_{i}a_{j}^* + \mathcal{T}_{j}a_{i}^*\right) + G_{ij}^* + \frac{\sigma_{ij} + \sigma_{ji}}{f_{ij}}. 
\end{align} 

A crucial result which we provide below establishes that the welfare under any intervention that subsidizes some agent's action can be strictly improved by a different intervention that is budget feasible. 

\begin{lemma} \label{l:welfareimproving}
  Consider any sequence of environments and corresponding sequence of zero-subsidy networks $\{\bm{G}^{(0)}_{m}\}_{m \in \mathbb{N}}$ satisfying \Cref{as:highspillovers}. For $m$ large enough, consider any subsidy $\left(\bm{\beta},\bm{\sigma}\right)$ that is budget-feasible and has $\beta_{i} > 0$ for some $i \in \mathcal{N}$. There exists another subsidy $\left(\widetilde{\bm{\beta}},\widetilde{\bm{\sigma}}\right)$ that is budget feasible and obtains strictly greater welfare than $\left(\bm{\beta},\bm{\sigma}\right)$. 
\end{lemma}

\begin{proof}
  Consider any agent $i$ with action subsidy $\beta_{i} > 0$. The main part of the proof establishes that when $m$ is large enough, there is some agent $j$ such that the quantity \begin{align} \label{eq:perturbexo}
       \underbrace{\frac{d \mathcal{W}}{d \sigma_{ij}} \cdot \frac{d \mathcal{P}}{d \beta_{i}} - \frac{d \mathcal{W}}{d \beta_{i}} \cdot \frac{d \mathcal{P}}{d \sigma_{ij}}}_{\Delta},
    \end{align} is the product of $\|\nabla \mathcal{W}\|_1$ and a  positive term. We establish this fact and then apply it to prove the lemma.

    By \Cref{as:highspillovers} Part (a), there must exist an agent $j$ such that $s_{ij} > 0$. We will fix this agent and analyze the link subsidy to him in the rest of the argument. Note that $g_{ij}^*>0$ and thus the welfare and planner payments are differentiable in $\sigma_{ij}$. Substituting (\ref{eq:bderivexo}) and (\ref{eq:sderivexo}) in (\ref{eq:perturbexo}), we have \begin{align} \label{eq:deltal}
        \Delta &= \frac{\rho}{f_{ij}} \left(\mathcal{S}_{i}a_{j}^* + \mathcal{S}_{j}a_{i}^*\right)\left(\mathcal{T}_{i} + a_{i}^*\right) - \mathcal{S}_{i} \left(\frac{\rho}{f_{ij}} \left(\mathcal{T}_{i}a_{j}^* + \mathcal{T}_{j}a_{i}^*\right) + G_{ij}^* + \frac{\sigma_{ij} + \sigma_{ji}}{f_{ij}}\right), \\
        & = \frac{\rho a_{i}^*}{f_{ij}} \left(\mathcal{S}_{j}\mathcal{T}_{i} + \mathcal{S}_{i}a_{j}^* + \mathcal{S}_{j}a_{i}^* - \mathcal{S}_{i}\mathcal{T}_{j}\right) - \mathcal{S}_{i} \left(G_{ij}^* + \frac{\sigma_{ij} + \sigma_{ji}}{f_{ij}}\right), \\
        & \geq \frac{\rho \left(a_{i}^*\right)^2 \mathcal{S}_{j}}{f_{ij}} - \mathcal{S}_{i} \left(G_{ij}^* + \frac{\sigma_{ij} + \sigma_{ji}}{f_{ij}}\right). \label{eq:deltalb}
    \end{align} Here, the final inequality is obtained by comparing (\ref{eq:eqlbexo}) and the definition of $\mathcal{T}_{i}$ in (\ref{eq:bderivexo}). To compare these terms, observe that since the spectral radius of $\rho \bm{G}^*$ is at most $1$ and $b_{i} \geq \underline{b}$, we have $a_{j}^* > \mathcal{T}_{j}$ (and similarly $a_{i}^* > \mathcal{T}_{i}$). Since $\mathcal{S}_{j} \geq 0$ and $\mathcal{T}_{i} \geq 0$, this implies the final inequality.
    
    We now obtain a lower bound on $\Delta$ by bounding the quantities on the right hand side. We will denote $\left(\bm{I}-\rho \bm{G}^*\right)^{-1}$ by $\bm{M}$. Because $\rho \mu_{1}(\bm{G}^*) < 1$, the matrix $\bm{M}$ is positive semidefinite with non-negative entries. A standard fact therefore implies $M_{kl} \leq \sqrt{M_{kk}M_{ll}}$ for any $k,l \in \mathcal{N}$. Furthermore, $M_{kk} \leq \|\bm{M}\|_{2}$ for any $k$ and thus $M_{kl} \leq 1/\left(1-\rho \mu_{1}(\bm{G}^*)\right)$ for any $k,l \in \mathcal{N}$. Substituting in the definition of $\mathcal{S}_{i}$ (see (\ref{eq:bderivexo})), we have \begin{align} \label{eq:ubs}
        \mathcal{S}_{i} = \sum_{l \in \mathcal{N}}M_{il}\nabla \mathcal{W}_{l} \leq \|\nabla \mathcal{W}\|_{1} \cdot \frac{1}{1-\rho \mu_{1}\left(\bm{G}^*\right)}.
    \end{align} The last inequality follows from the fact that links under any non-negative subsidy are weakly greater than links under the zero subsidy and thus $\mu_{1}\left(\bm{G}^{(0)}\right) \leq \mu_{1} \left(\bm{G}^*\right)$. We now establish a lower bound on $\mathcal{S}_{j}$. By \Cref{as:highspillovers} Part (a), consider the following chain of inequalities where the inequality is pointwise on the elements of the matrix: \begin{align*}
        \bm{M} \geq \left(\rho \bm{G}^*\right)^{L} \geq \left(\rho \bm{G}^{(0)}\right)^{L} \geq \delta \left(\rho \mu_{1}\left(\bm{G}^{(0)}\right)\right)^{L} \bm{1} \bm{1}^{\top}.
    \end{align*} Thus, \begin{align*}
        \mathcal{S}_{j} = \sum_{l \in \mathcal{N}}M_{jl}\nabla \mathcal{W}_{l} \geq \|\nabla \mathcal{W}\|_{1} \cdot  \delta \left(\rho \mu_{1}\left(\bm{G}^{(0)}\right)\right)^{L}. 
    \end{align*} Next, we establish a lower bound on equilibrium actions. Since $\rho \mu_{1}\left(\bm{G}^*\right) < 1$ and by  \Cref{as:highspillovers} Part (a), we have $\bm{a}^* \geq \left(\rho\bm{G}^*\right)^{L} \bm{a}^* \geq \delta \left(\rho \mu_{1}\left(\bm{G}^{(0)}\right)\right)^{L} \bm{1} \bm{1}^{\top} \bm{a}^*$. Now, equilibrium actions satisfy $ \bm{a}^* = \bm{M}\left(\bm{b} + \bm{\beta}\right)$. Premultiplying each side by the top-most eigenvector of $\bm{M}$, we have $\bm{v}_{1}^{\top}\bm{a}^* = \bm{v}_{1}^{\top} \left(\bm{b} + \bm{\beta}\right)/\left(1-\rho \mu_{1}\left(\bm{G}^*\right)\right)$. Since $\|\bm{v}_{1}\|_{2} = 1$, we have $\bm{1}^{\top}\bm{a}^* \geq \bm{v}_{1}^{\top}\left(\bm{b} + \bm{\beta}\right)/\left(1-\rho \mu_{1}\left(\bm{G}^{*}\right)\right)$. Using $b_{i} \geq \underline{b}$ for every $i$, we have \begin{align*}
        \bm{a}^* \geq \delta \left(\rho \mu_{1}\left(\bm{G}^{(0)}\right)\right)^{L} \cdot \frac{\underline{b}}{1-\rho \mu_{1}\left(\bm{G}^{*}\right)} \bm{1}.
    \end{align*}

    Finally, we establish an upper bound on $G_{ij}^*$ and the subsidy $\sigma_{ij}$ and $\sigma_{ji}$. Note that because $s_{ij} > 0$ and $s_{ij}=s_{ji}$, we have $G_{ij}^* \geq \left(\sigma_{ij} + \sigma_{ji}\right)/f_{ij}$. Since the intervention is budget feasible, we must have that $\left(\sigma_{ij}+\sigma_{ji}\right)G_{ij}^* \leq \mathcal{B}$. Thus, $\sigma_{ij} + \sigma_{ji} \leq \sqrt{\mathcal{B} f_{ij}}$. Now, we derive an upper bound on $G_{ij}^*$. Because the row sum of the matrix is upper bounded by the spectral radius, we have $ \rho s_{ij} \leq f_{ij}$ by \Cref{as:highspillovers} Part (b). Furthermore, using $f_{ij} \leq \overline{f}$, the baseline link incentive $s_{ij} \leq \overline{f} / \rho$. Using the characterization of equilibrium links, we have $G_{ij}^* \leq \left(2\frac{\overline{f}}{\rho} + \sigma_{ij} + \sigma_{ji}\right)/f_{ij} \leq \left(2\frac{\overline{f}}{\rho} + \sqrt{\mathcal{B} \overline{f}}\right)/\underline{f}$.

    Substituting in the lower bound on $\Delta$ established in (\ref{eq:deltalb}), we have \begin{multline}
        \Delta \geq \|\nabla \mathcal{W}\|_{1} \cdot \\ \left(\frac{\rho \delta^{3} \underline{b}^{2}}{\overline{f}} \left(\rho \mu_{1}\left(\bm{G}^{(0)}\right)\right)^{3L} \frac{1}{\left(1-\rho \mu_{1}\left(\bm{G}^{*}\right)\right)^2} - \frac{1}{1-\rho \mu_{1}\left(\bm{G}^{*}\right)} \cdot \left(\frac{2\frac{\overline{f}}{\rho} + \sqrt{\mathcal{B} \overline{f} }}{\underline{f}} + \sqrt{\frac{\mathcal{B} \overline{f}}{\underline{f}}}\right) \right).
    \end{multline} Note that $\rho \mu_{1}(\bm{G}^{(0)}) \leq \rho \mu_{1}\left(\bm{G}^*\right) < 1$. By \Cref{as:highspillovers}   (b) and the facts that the budget is bounded above and $\rho$ is bounded below, the coefficient of $ \|\nabla \mathcal{W}\|_{1}$ is positive for $m$ sufficiently large.

    If $ \|\nabla \mathcal{W}\|_{1} > 0$, then we conclude that $\Delta > 0$. Then there is a budget-neutral deviation for the planner to increase $\sigma_{ij}$ and decrease $\beta_i$ that increases welfare. If $ \|\nabla \mathcal{W}\|_{1} = 0$, since $\mathcal{W}(\cdot)$ is increasing in each argument, by continuity we can also find such a deviation.
\end{proof}

Now, consider any optimal subsidy $\left(\bm{\beta}^*,\bm{\sigma}^*\right)$ for $m$ large enough and suppose $\beta_{i}^*>0$. By \Cref{l:welfareimproving}, there exists another budget-feasible subsidy which gives the planner strictly greater welfare. This contradicts the optimality of $\left(\bm{\beta}^*,\bm{\sigma}^*\right)$. Thus, we can conclude that $\beta_{i}^*=0$ for every agent $i \in \mathcal{N}$.
\end{proof}

\begin{proof}[Proof of \Cref{t:baselineincentivegeneral}]
    The proof follows an identical approach to the proof of \Cref{t:endogenousmodel}. We assume the link intensity constraint between $i$ and $j$ binds. This is without loss when $s_{ij} \geq 0$ and is assumed when $s_{ij} < 0$. 
    
    For every agent $j \in \mathcal{N}$, we have \begin{align} \label{eq:actionfocgeneral}
        C_{j}'(a_{j}^*) = b_{j} + \beta_{j} + \rho \sum_{k \in \mathcal{N} \setminus \{j\}} G_{jk}^* u_{jk}'(a_{j}^*)u_{kj}(a_{k}^*)
    \end{align} Furthermore, \begin{align}
        \label{eq:linkfocgeneral}
        F_{kl}'(g_{kl}^*) = \sigma_{kl} + s_{kl} + \rho u_{kl}(a_{k}^*) u_{lk}(a_{l}^*)
    \end{align} for every pair of individuals $k \in \mathcal{N}$ and $l \in \mathcal{N}$ for whom the link incentive constraint binds.  We will implicitly differentiate (\ref{eq:actionfocgeneral}) and (\ref{eq:linkfocgeneral}) with respect to $\beta_{i}$ and $\sigma_{ij}$  to solve for the derivative of equilibrium actions. For the rest of the proof, we assume the link intensity constraint for every pair of agents binds. An argument identical to that in the proof of \Cref{t:endogenousmodel} shows that this is without loss. We begin by differentiating with respect to the action subsidy $\beta_{i}$. For every agent $j \in \mathcal{N} \setminus \{i\}$, implicitly differentiating (\ref{eq:actionfocgeneral}) with respect to $\beta_{i}$, we have \begin{multline}
        \label{eq:actionderbetajneqigeneral}
        C_{j}''(a_{j}^*) \frac{d a_{j}^*}{d \beta_{i}} = \rho \sum_{k \in \mathcal{N} \setminus \{j\}}G_{jk}^*\left(u_{jk}''(a_{j}^*)u_{kj}(a_{k}^*)\frac{d a_{j}^*}{d \beta_{i}} + u_{jk}'(a_{j}^*)u_{kj}'(a_{k}^*)\frac{d a_{k}^*}{d \beta_{i}}\right) \\ + \rho \sum_{k \in \mathcal{N} \setminus \{j\}}\frac{d G_{jk}^*}{d \beta_{i}} u_{jk}'(a_{j}^*) u_{kj}\left(a_{k}^*\right), 
    \end{multline} and \begin{multline}
        \label{eq:actionderbetajeqigeneral}
        C_{i}''(a_{i}^*) \frac{d a_{i}^*}{d \beta_{i}} = 1 + \rho \sum_{k \in \mathcal{N} \setminus \{i\}}G_{ik}^*\left(u_{ik}''(a_{i}^*)u_{ki}(a_{k}^*)\frac{d a_{i}^*}{d \beta_{i}} + u_{ik}'(a_{i}^*)u_{ki}'(a_{k}^*)\frac{d a_{k}^*}{d \beta_{i}}\right) \\ + \rho \sum_{k \in \mathcal{N} \setminus \{i\}}\frac{d G_{ik}^*}{d \beta_{i}} u_{ik}'(a_{i}^*) u_{ki}\left(a_{k}^*\right).
    \end{multline} Moreover, implicitly differentiating (\ref{eq:linkfocgeneral}) with respect to $\beta_{i}$, for every link intensity, we have \begin{align}
        \label{eq:linkderivbetageneral}
        F_{kl}''(g_{kl}^*) \frac{d g_{kl}^*}{d \beta_{i}} = \rho \left(u_{kl}'(a_{k}^*)u_{lk}(a_{l}^*)\frac{d a_{k}^*}{d \beta_{i}} + u_{kl}(a_{k}^*)u'_{lk}(a_{l}^*)\frac{d a_{l}^*}{d \beta_{i}}\right). 
    \end{align} Substitute (\ref{eq:linkderivbetageneral}) in (\ref{eq:actionderbetajneqigeneral}) and (\ref{eq:actionderbetajeqigeneral}) to solve the system of equations in variables $\left(\frac{d a_{k}^*}{d \beta_{i}}\right)_{k \in \mathcal{N}}$. Consequently, there is a matrix $\bm{\mathcal{M}}$ such that \begin{align*}
        \left(\bm{C}-\rho \bm{\mathcal{M}}\right)\frac{d \bm{a}^*}{d \beta_{i}} = \bm{e}_{i},
    \end{align*} where $\bm{C}$ is a diagonal matrix with $C_{ii} = C_{i}''(a_{i}^*)$, and $\bm{e}_{i}$ is the $i^{\text{th}}$ basis vector. By assumption, equilibrium $(\bm{a}^*,\bm{g}^*)$ is differentiable at subsidy $(\bm{\beta},\bm{\sigma})$ and thus \begin{align*}
        \frac{d \bm{a}^*}{d \beta_{i}} = \left(\bm{C}-\rho \bm{\mathcal{M}}\right)^{-1}_{i},
    \end{align*} where $\left(\bm{C}-\rho \bm{\mathcal{M}}\right)^{-1}_{i}$ is the $i^{\text{th}}$ column of matrix $\left(\bm{C}-\rho \bm{\mathcal{M}}\right)^{-1}$. 

We now follow the same set of calculations to solve for the derivative of equilibrium actions $\bm{a}^*$ with respect to link subsidy $\sigma_{ij}$. For every agent $k \in \mathcal{N}$, implicitly differentiating (\ref{eq:actionfocgeneral}) with respect to $\sigma_{ij}$, we have \begin{multline}
        \label{eq:actiondersigmageneral}
        C_{k}''(a_{k}^*) \frac{d a_{k}^*}{d \sigma_{ij}} = \rho \sum_{l \in \mathcal{N} \setminus \{k\}}G_{kl}^*\left(u_{kl}''(a_{k}^*)u_{lk}(a_{l}^*)\frac{d a_{k}^*}{d \sigma_{ij}} + u_{kl}'(a_{k}^*)u_{lk}'(a_{l}^*)\frac{d a_{l}^*}{d \sigma_{ij}}\right) \\ + \rho \sum_{l \in \mathcal{N} \setminus \{k\}}\frac{d G_{kl}^*}{d \sigma_{ij}} u_{kl}'(a_{k}^*) u_{lk}\left(a_{l}^*\right).
    \end{multline} Moreover, implicitly differentiating (\ref{eq:linkfocgeneral}) with respect to $\sigma_{ij}$, for every link intensity from agent $k$ to $l$, excluding that from agent $i$ to $j$, we have \begin{align} \label{eq:linkderivsigmageneral}
        F_{kl}''(g_{kl}^*) \frac{d g_{kl}^*}{d \sigma_{ij}} = \rho \left(u_{kl}'(a_{k}^*)u_{lk}(a_{l}^*)\frac{d a_{k}^*}{d \sigma_{ij}} + u_{kl}(a_{k}^*)u'_{lk}(a_{l}^*)\frac{d a_{l}^*}{d \sigma_{ij}}\right),
    \end{align} and for the link intensity from agent $i$ to $j$, we have \begin{align} \label{eq:linkderivsigmageneraleqi}
        F_{ij}''(g_{ij}^*) \frac{d g_{ij}^*}{d \sigma_{ij}} = 1+ \rho \left(u_{ij}'(a_{i}^*)u_{ji}(a_{j}^*)\frac{d a_{i}^*}{d \sigma_{ij}} + u_{ij}(a_{i}^*)u'_{ji}(a_{j}^*)\frac{d a_{j}^*}{d \sigma_{ij}}\right).
    \end{align} Substitute (\ref{eq:linkderivsigmageneraleqi}) and (\ref{eq:linkderivsigmageneral}) in (\ref{eq:actiondersigmageneral}) to solve the system of equations, we have \begin{align*}
        \left(\bm{C}-\rho \bm{\mathcal{M}}\right) \frac{d \bm{a}^*}{d \sigma_{ij}} = \rho \frac{1}{F_{ij}''(g_{ij}^*)} \left(u_{ij}'(a_{i}^*)u_{ji}(a_{j}^*)\bm{e}_{i} + u_{ij}(a_{i}^*)u_{ji}'(a_{j}^*)\bm{e}_{j}\right),
    \end{align*} where $\bm{e}_{i}$ and $\bm{e}_{j}$ are standard basis vectors. By assumption, equilibrium $(\bm{a}^*,\bm{g}^*)$ is differentiable at subsidy $(\bm{\beta},\bm{\sigma})$ and thus \begin{align} \label{eq:actderivsigmageneral}
        \frac{d \bm{a}^*}{d \sigma_{ij}} = \rho \underbrace{\frac{1}{F_{ij}''(g_{ij}^*)} }_{\widetilde{f}_{ij}}\left(u_{ij}'(a_{i}^*)u_{ji}(a_{j}^*)\left(\bm{C}-\rho \bm{\mathcal{M}}\right)^{-1}_{i} + u_{ij}(a_{i}^*)u'_{ji}(a_{j}^*) \left(\bm{C}-\rho \bm{\mathcal{M}}\right)^{-1}_{j}\right). 
    \end{align}

 For a given welfare objective $\mathcal{W}(\cdot)$, the Lagrangian of the planner's optimization problem is given by \begin{align*}
        \mathcal{L}(\bm{a},\bm{g},\bm{\beta},\bm{\sigma},\lambda) = \mathcal{W}(\bm{a}) - \lambda \left(\sum_{k \in \mathcal{N}}\beta_{k}a_{k} + \sum_{k \in \mathcal{N}}\sum_{l \in \mathcal{N} \setminus \{k\}}\sigma_{kl}G_{kl}\right).
    \end{align*} We begin by taking the derivative of the Lagrangian with respect to $\beta_{i}$: \begin{align} \label{eq:lagderivbetageneral}
        \frac{d \mathcal{L}}{d \beta_{i}} = \nabla \mathcal{W}^{\top} \frac{d \bm{a}^*}{d \beta_{i}} -\lambda \left(\sum_{k \in \mathcal{N}} \beta_{k} \frac{d a_{k}^*}{d \beta_{i}} + \sum_{k \in \mathcal{N}}\sum_{l \in \mathcal{N} \setminus \{k\}}\sigma_{kl}\frac{d G_{kl}^*}{d \beta_{i}}\right) - \lambda a_{i}^*.
    \end{align} Substituting (\ref{eq:linkderivbetageneral}) in the expression above, there is a vector $\bm{\omega}$ such that the derivative of the Lagrangian with respect to $\beta_{i}$ satisfies \begin{align} \label{eq:reducedvecgeneral}
        \frac{d \mathcal{L}}{d \beta_{i}} = \bm{\omega}^{\top} \frac{d \bm{a}^*}{d \beta_{i}} - \lambda a_{i}^*.
    \end{align} Furthermore, substituting the expression \begin{align*}
        \frac{d \bm{a}^*}{d \beta_{i}} = \left(\bm{C}-\rho \bm{\mathcal{M}}\right)^{-1}_{i},
    \end{align*} into (\ref{eq:reducedvecgeneral}) we have \begin{align}
        \frac{d \mathcal{L}}{d \beta_{i}} &= \bm{\omega}^{\top} \left(\bm{C}-\rho \bm{\mathcal{M}}\right)^{-1}_{i}  - \lambda a_{i}^*, \\ \label{eq:betafocfingeneral}
        &= \mathcal{R}_{i}' - \lambda a_{i}^*,
    \end{align} where we define $\mathcal{R}_{i}' := \bm{\omega}^{\top} \left(\bm{C}-\rho \bm{\mathcal{M}}\right)^{-1}_{i}$.

We next take the derivative of the Lagrangian with respect to $\sigma_{ij}$: \begin{align} \label{eq:lagderivsigmageneral}
        \frac{d \mathcal{L}}{d \sigma_{ij}} = \nabla \mathcal{W}^{\top} \frac{d \bm{a}^*}{d \sigma_{ij}} -\lambda \left(\sum_{k \in \mathcal{N}} \beta_{k} \frac{d a_{k}^*}{d \sigma_{ij}} + \sum_{k \in \mathcal{N}}\sum_{l \in \mathcal{N} \setminus \{k\}}\sigma_{kl}\frac{d G_{kl}^*}{d \sigma_{ij}}\right) - \lambda G_{ij}^*.
    \end{align} Substituting (\ref{eq:actderivsigmageneral}) into (\ref{eq:lagderivsigmageneral}), we have \begin{align}
        \frac{d \mathcal{L}}{d \sigma_{ij}} &= \bm{\mathcal{\omega}}^{\top} \frac{d \bm{a}^*}{d \sigma_{ij}} - \lambda \widetilde{f}_{ij}\left(\sigma_{ij}+\sigma_{ji}\right) - \lambda G_{ij}^*, \\ \label{eq:lagderivsigmageneraltwo}
        &= \rho \widetilde{f}_{ij} \left(\mathcal{R}_{i}'u_{ij}'(a_{i}^*)u_{ji}(a_{j}^*) + \mathcal{R}_{j}'u_{ij}(a_{i}^*)u'_{ji}(a_{j}^*)\right) - \lambda \widetilde{f}_{ij}\left(\sigma_{ij} + \sigma_{ji}\right) - \lambda G_{ij}^*. 
    \end{align}

An analogous calculation will establish that \begin{align}
    \frac{d \mathcal{L}}{d \sigma_{ji}} &= \rho \widetilde{f}_{ji} \left(\mathcal{R}_{j}'u_{ji}'(a_{j}^*)u_{ij}(a_{i}^*) + \mathcal{R}_{i}'u_{ji}(a_{j}^*)u_{ij}'(a_{i}^*)\right) - \lambda \widetilde{f}_{ji} \left(\sigma_{ij}+\sigma_{ji}\right) - \lambda G_{ij}^*. \label{eq:lagderivsigmageneralji} 
\end{align}

We will now establish optimal subsidies by solving for the KKT first-order conditions. Instead of analyzing the derivative of the objective with respect to each individual link intensity $\sigma_{ij}$, we analyze the sum of the derivatives across pairs of agents. Adding (\ref{eq:lagderivsigmageneralji}) and (\ref{eq:lagderivsigmageneraltwo}), we get

\begin{multline} \label{eq:sigmaderivsum}
    \underbrace{\frac{\partial \mathcal{L}}{\partial \sigma_{ij}} + \frac{\partial \mathcal{L}}{\partial \sigma_{ji}}}_{\Delta_{ij}\mathcal{L}} = \rho \left(\widetilde{f}_{ij} + \widetilde{f}_{ji}\right) \left(\mathcal{R}_{i}' u_{ij}'(a_{i}^*) u_{ji}(a_{j}^*) + \mathcal{R}_{j}' u_{ij}(a_{i}^*) u_{ji}'(a_{j}^*)\right) \\ - \left(\widetilde{f}_{ij} + \widetilde{f}_{ji}\right)  \lambda \left(\sigma_{ij} + \sigma_{ji}\right) - 2\lambda G_{ij}^*. 
\end{multline}

Consider any optimal subsidy $(\bm{\beta}^*,\bm{\sigma}^*)$ and corresponding equilibrium $(\bm{a}^*,\bm{g}^*)$ such that the budget constraint binds, that is, $\lambda > 0$. Analyzing (\ref{eq:betafocfingeneral}), the KKT first-order conditions imply \begin{align*}
    \mathcal{R}_{i}' \leq \lambda a_{i}^* \quad \text{ and } \mathcal{R}_{j}' \leq \lambda a_{j}^*
\end{align*} Substituting in (\ref{eq:sigmaderivsum}), we have \begin{align*}
    \Delta_{ij}\mathcal{L} &\leq \lambda  \left(\widetilde{f}_{ij} + \widetilde{f}_{ji}\right) \left(\rho a_{i}^* u_{ij}'(a_{i}^*) u_{ji}(a_{j}^*) + \rho a_{j}^*u_{ij}(a_{i}^*)u_{ji}'(a_{j}^*) - \frac{2G_{ij}^*}{\widetilde{f}_{ij} + \widetilde{f}_{ji}} - \sigma_{ij} - \sigma_{ji}\right).
\end{align*} Suppose the cost $F_{ij}(\cdot)$ is super-quadratic. Then, this equation can be bounded and rearranged as \begin{multline*}
   \frac{\Delta_{ij} \mathcal{L}}{\lambda  \left(\widetilde{f}_{ij} + \widetilde{f}_{ji}\right)} \leq  \rho a_{i}^* u_{ij}'(a_{i}^*) u_{ji}(a_{j}^*) + \rho a_{j}^*u_{ij}(a_{i}^*)u_{ji}'(a_{j}^*) - 2F_{ij}'(g_{ij}^*) \frac{F_{ij}''(g_{ji}^*)}{F_{ij}''(g_{ji}^*)+F_{ij}''(g_{ij}^*)}  \\ - 2F_{ij}'(g_{ji}^*) \frac{F_{ij}''(g_{ij}^*)}{F_{ij}''(g_{ji}^*)+F_{ij}''(g_{ij}^*)} - \sigma_{ij} - \sigma_{ji},
\end{multline*} where we substituted $\widetilde{f}_{ij} + \widetilde{f}_{ji} = 1/F_{ij}''(g_{ij}^*) + 1/F_{ij}''(g_{ji}^*)$. Substituting the expression of equilibrium link intensities in (\ref{eq:linkfocgeneral}), we have \begin{multline*}
   \frac{\Delta_{ij} \mathcal{L}}{\lambda  \left(\widetilde{f}_{ij} + \widetilde{f}_{ji}\right)} \leq  \rho a_{i}^* u_{ij}'(a_{i}^*) u_{ji}(a_{j}^*) + \rho a_{j}^*u_{ij}(a_{i}^*)u_{ji}'(a_{j}^*) - 2s_{ij} -2\rho u_{ij}(a_{i}^*)u_{ji}(a_{j}^*) \\  - \sigma_{ij} \left(\frac{2F_{ij}''(g_{ji}^*)}{F_{ij}''(g_{ji}^*) + F_{ij}''(g_{ij}^*)} + 1\right) - \sigma_{ji} \left(\frac{2F_{ij}''(g_{ij}^*)}{F_{ij}''(g_{ji}^*) + F_{ij}''(g_{ij}^*)}+1\right). 
\end{multline*} When $u_{ij}(\cdot)$ and $u_{ji}(\cdot)$ are sub-linear, we have \begin{align} \label{eq:deltalagcondition}
    \frac{\Delta_{ij} \mathcal{L}}{\lambda  \left(\widetilde{f}_{ij} + \widetilde{f}_{ji}\right)} \leq -2s_{ij} - \sigma_{ij} \left(\frac{2F_{ij}''(g_{ji}^*)}{F_{ij}''(g_{ji}^*) + F_{ij}''(g_{ij}^*)} + 1\right) - \sigma_{ji} \left(\frac{2F_{ij}''(g_{ij}^*)}{F_{ij}''(g_{ji}^*) + F_{ij}''(g_{ij}^*)}+1\right).
\end{align} 

We will now show that at an optimal intervention $(\bm{\beta},\bm{\sigma})$, the link subsidy $\sigma_{ij}^*=\sigma_{ji}^*=0$ when $s_{ij} \geq 0$. To establish a contradiction, suppose
$\sigma_{ij}^*>0$ and $\sigma_{ji}^*>0$. By the KKT first-order conditions, it must be that the derivative of the Lagrangian with respect to $\sigma_{ij}$ and $\sigma_{ji}$ is zero and thus $\Delta_{ij} \mathcal{L} = 0$. This is a contradiction to the inequality in (\ref{eq:deltalagcondition}). 

Next, suppose $\sigma_{ij}^*>0$ and $\sigma_{ji}^*=0$. Analyzing the equilibrium link intensity condition in (\ref{eq:linkfocgeneral}), because $\sigma_{ij}^* > \sigma_{ji}^*$ we have $g_{ij}^* > g_{ji}^*$. By assumption that the derivative of the cost of link intensity $F_{ij}'(\cdot)$ is weakly convex, we have \begin{align*}
    \widetilde{f}_{ji} = \frac{1}{F_{ij}''(g_{ji}^*)} \geq \frac{1}{F_{ij}''(g_{ij}^*)} = \widetilde{f}_{ij}. 
\end{align*} We now compare the derivative of the Lagrangian with respect to $\sigma_{ij}$ and $\sigma_{ji}$. By assumption that $\sigma_{ij}^*>0$, the KKT first-order conditions gives $d \mathcal{L} / d \sigma_{ij} = 0$. Substituting in (\ref{eq:lagderivsigmageneraltwo}) we have \begin{align*}
    \widetilde{f}_{ij} \underbrace{\left(\rho \left(\mathcal{R}_{i}'u_{ij}'(a_{i}^*)u_{ji}(a_{j}^*) + \mathcal{R}_{j}'u_{ij}(a_{i}^*)u_{ji}'(a_{j}^*)\right) - \lambda \left(\sigma_{ij}^*+\sigma_{ji}^*\right) \right)}_{\mathcal{T}_{ij}}  = \lambda G_{ij}^*. 
\end{align*} Since $g_{ij}^*>0$, the overall link $G_{ij}^*>0$ and consequently $\mathcal{T}_{ij}>0$. Analyzing the derivative of the Lagrangian with respect to $\sigma_{ji}$ in (\ref{eq:lagderivsigmageneralji}), we have \begin{align*}
    \frac{d \mathcal{L}}{d \sigma_{ji}} = \widetilde{f}_{ji} \mathcal{T}_{ij} - \lambda G_{ij}^* \geq \widetilde{f}_{ij} \mathcal{T}_{ij} - \lambda G_{ij}^*  = \frac{d \mathcal{L}}{d \sigma_{ij}}.
\end{align*} Here, the inequality follows from the fact that $\mathcal{T}_{ij} > 0$ and $\widetilde{f}_{ji} \geq \widetilde{f}_{ij}$. Using $\frac{d \mathcal{L}}{d \sigma_{ij}}=0$, we have $d \mathcal{L}/d \sigma_{ji} \geq 0$ and thus $\Delta_{ij} \mathcal{L} \geq 0$. But this is a contradiction to the bound derived in (\ref{eq:deltalagcondition}). A symmetric analysis establishes that it cannot be the case $\sigma_{ji}^*>0$ and $\sigma_{ij}^*=0$. Thus, it must be $\sigma_{ij}^*=\sigma_{ji}^*=0$.

\textit{Proof of Part (ii)}: Suppose $\beta_{i}^*>0$ and $\beta_{j}^*>0$. By the KKT first-order conditions, $\mathcal{R}_{i}' = \lambda a_{i}^*$ and $\mathcal{R}_{j}' = \lambda a_{j}^*$. Substituting in (\ref{eq:sigmaderivsum}), we have \begin{align*}
    \Delta_{ij} \mathcal{L} = \lambda  \left(\widetilde{f}_{ij} + \widetilde{f}_{ji} \right)\left(\rho \left(a_{i}^* u_{ij}'(a_{i}^*) u_{ji}(a_{j}^*) + a_{j}^*u_{ij}(a_{i}^*)u_{ji}'(a_{j}^*) \right) - \frac{2G_{ij}^*}{\widetilde{f}_{ij} + \widetilde{f}_{ji}} - \sigma_{ij} - \sigma_{ji}\right).
\end{align*} Suppose the cost $F_{ij}(\cdot)$ is sub-quadratic. Then, the equation can be bounded and rearranged as \begin{multline*}
   \Delta_{ij} \mathcal{L} \cdot \frac{1}{\lambda  \left(\widetilde{f}_{ij} + \widetilde{f}_{ji}\right)} \geq  \rho a_{i}^* u_{ij}'(a_{i}^*) u_{ji}(a_{j}^*) + \rho a_{j}^*u_{ij}(a_{i}^*)u_{ji}'(a_{j}^*) - 2F_{ij}'(g_{ij}^*) \frac{F_{ij}''(g_{ji}^*)}{F_{ij}''(g_{ji}^*)+F_{ij}''(g_{ij}^*)}  \\ - 2F_{ij}'(g_{ji}^*) \frac{F_{ij}''(g_{ij}^*)}{F_{ij}''(g_{ji}^*)+F_{ij}''(g_{ij}^*)} - \sigma_{ij} - \sigma_{ji},
\end{multline*} where we substituted $\widetilde{f}_{ij} + \widetilde{f}_{ji} = 1/F_{ij}''(g_{ij}^*) + 1/F_{ij}''(g_{ji}^*)$. Substituting the expression of equilibrium link first-order conditions (\ref{eq:linkfocgeneral}) above, we have \begin{multline}
   \frac{\Delta_{ij} \mathcal{L}}{\lambda  \left(\widetilde{f}_{ij} + \widetilde{f}_{ji}\right)} \geq  \rho a_{i}^* u_{ij}'(a_{i}^*) u_{ji}(a_{j}^*) + \rho a_{j}^*u_{ij}(a_{i}^*)u_{ji}'(a_{j}^*) - 2s_{ij} -2\rho u_{ij}(a_{i}^*)u_{ji}(a_{j}^*) \\  - \sigma_{ij}\left(\frac{2F_{ij}''(g_{ji}^*)}{F_{ij}''(g_{ji}^*)+F_{ij}''(g_{ij}^*)}+1\right) - \sigma_{ji}\left(\frac{2F_{ij}''(g_{ij}^*)}{F_{ij}''(g_{ji}^*)+F_{ij}''(g_{ij}^*)}+1\right). 
\end{multline} If $u_{ij}(\cdot)$ and $u_{ji}(\cdot)$ are super-linear, then \begin{align*}
     \frac{\Delta_{ij} \mathcal{L}}{\lambda  \left(\widetilde{f}_{ij} + \widetilde{f}_{ji}\right)} \geq -2s_{ij} - \sigma_{ij}\left(\frac{2F_{ij}''(g_{ji}^*)}{F_{ij}''(g_{ji}^*)+F_{ij}''(g_{ij}^*)}+1\right) - \sigma_{ji}\left(\frac{2F_{ij}''(g_{ij}^*)}{F_{ij}''(g_{ji}^*)+F_{ij}''(g_{ij}^*)}+1\right).
\end{align*} If $s_{ij} < 0$ and $\sigma_{ij}^*=\sigma_{ji}^*=0$, then by the above equation $\Delta_{ij} \mathcal{L} > 0$. This contradicts the KKT first-order conditions which imply $\Delta_{ij} \mathcal{L} \leq 0$. Thus, we must have $\sigma_{ij}^*+\sigma_{ji}^* > 0$. 
\end{proof}

\end{document}